\documentclass[11pt, reqno]{amsart}

\usepackage{amsmath, amsfonts, amssymb, amsbsy, bigstrut, graphicx, enumerate,  upref, longtable, comment}

\usepackage{pstricks}

\usepackage[authoryear, round]{natbib} 

\usepackage[breaklinks]{hyperref}

\newcommand{\vol}{\mathrm{Vol}}
\newcommand{\I}{\mathrm{i}}

\newcommand{\ep}{\epsilon}

\newtheorem{thm}{Theorem}[section]
\newtheorem{lmm}[thm]{Lemma}
\newtheorem{cor}[thm]{Corollary}
\newtheorem{prop}[thm]{Proposition}

\theoremstyle{remark}

\newcommand{\cc}{\mathbb{C}}
\newcommand{\cov}{\mathrm{Cov}}
\newcommand{\dd}{\mathcal{D}}

\newcommand{\ee}{\mathbb{E}}

\newcommand{\ma}{\mathcal{A}}

\newcommand{\pp}{\mathbb{P}}

\newcommand{\ra}{\rightarrow}
\newcommand{\rr}{\mathbb{R}}

\newcommand{\tr}{\operatorname{Tr}}

\newcommand{\zz}{\mathbb{Z}}

\newcommand{\fpar}[2]{\frac{\partial #1}{\partial #2}}

\newcommand{\ftw}{\mathbb{T}}

\newcommand{\mb}{\mathcal{B}}

\numberwithin{equation}{section}

\renewcommand{\Re}{\operatorname{Re}}

\begin{document}
\title{The leading term of the Yang--Mills free energy}
\author{Sourav Chatterjee}
\address{\newline Department of Statistics \newline Stanford University\newline Sequoia Hall, 390 Serra Mall \newline Stanford, CA 94305\newline \newline \textup{\tt souravc@stanford.edu}}
\thanks{Research partially supported by NSF grant DMS-1441513}

\begin{abstract}
This article gives an explicit formula for the leading term of the free energy of three-dimensional $U(N)$ lattice gauge theory for any $N$, as the lattice spacing tends to zero. The proof is based on a novel technique that avoids phase cell renormalization. The technique also yields a similar formula for the four-dimensional theory, but only in the weak coupling~limit.
\end{abstract}

\maketitle


\section{Introduction}\label{intro}
Quantum Yang--Mills theories, also called quantum gauge theories, are the basic components of the Standard Model of quantum mechanics. Lattice gauge theories are discrete approximations of quantum Yang--Mills theories. Although lattice gauge theories are well-defined objects, the rigorous mathematical construction of three- and four-dimensional quantum Yang--Mills theories in the continuum is still an unsolved problem. The solution of this problem has been a long-standing goal of the program of constructive quantum field theory, initiated in the Fifties and Sixties. The program saw enormous advances over the years, culminating in the Eighties with Tadeusz Ba\l aban's monumental proof of the ultraviolet stability of three- and four-dimensional lattice gauge theories. 



The main result of this article gives a formula for the leading term of the free energy of three-dimensional $U(N)$ Yang--Mills theory. A similar formula is also obtained in the weak coupling limit in dimension four.  The proof of the main result involves an interesting interplay of random matrix theory, Selberg-type integrals, properties of Gaussian measures and bare-hands probability theory.

  

The next section introduces the mathematical model of lattice gauge theories in a way that should be fully accessible to the casual reader. The results are presented immediately after that. A more detailed introduction that is geared towards the general audience, with an extensive discussion of the background and references to the literature, is given in Section \ref{litsec}.

\section{Results}\label{results}

We will begin with a general result about the weak coupling limit of lattice gauge theories. The results about three- and four-dimensional Yang--Mills theories will be subsequently stated as corollaries of this general result. 

Fix two integers $d\ge 2$ and $N\ge 1$. Let $\zz^d$ be the $d$-dimensional integer lattice. Let $e_1,\ldots, e_d$ denote the standard basis vectors of $\rr^d$. Let $U(N)$ be the group of $N\times N$ unitary matrices.  Let $I$ denote the $N\times N$ identity matrix. Define a function $\phi:U(N) \ra \rr$ as
\begin{align}\label{phidef}
\phi(U):= \Re(\tr(I-U))\,.
\end{align}
Let $\Lambda$ be a finite subset of $\zz^d$. Suppose that for any two adjacent  
vertices $x,y\in \Lambda$, we have a unitary matrix $U(x,y)\in U(N)$, with the constraint that $U(y,x)=U(x,y)^{-1}$ for all pairs of adjacent $x,y$. Any such assignment of unitary matrices to edges will be called a configuration. The set of all configurations will be denoted by $U(\Lambda)$. Let $\Lambda'$ be the set of triples $(x,j,k)$ such that $1\le j<k\le d$ and the vertices $x$, $x+e_j$, $x+e_k$ and $x+e_j+e_k$ all belong to $\Lambda$. The elements of $\Lambda'$ will be called plaquettes. The plaquette $(x,j,k)$ may also be viewed as a square in $\zz^d$, with vertices $x$, $x+e_j$, $x+e_k$ and $x+e_j+e_k$. A visual representation of a plaquette is given in Figure~\ref{plaquettefig}. For each plaquette $(x,j,k)\in \Lambda'$ and each configuration $U\in U(\Lambda)$, define
\[
U(x,j,k) := U(x,x+e_j)U(x+ e_j, x+e_j+e_k) U(x+e_j+e_k, x+e_k)U(x+e_k, x)\,.
\]
Given a configuration $U\in U(\Lambda)$, define the Wilson action
\begin{align}\label{hdef}
S_\Lambda(U) := \sum_{(x,j,k)\in \Lambda'} \phi(U(x,j,k))\,,
\end{align}
where $\phi$ is defined in \eqref{phidef}. 
Let $\sigma_{\Lambda}$ be the product Haar measure on $U(\Lambda)$. For any $g_0>0$, let $\mu_{\Lambda, g_0}$ be the probability measure on $U(\Lambda)$ defined as
\[
d\mu_{\Lambda, g_0}(U) := \frac{1}{Z(\Lambda, g_0)}\exp\biggl(- \frac{1}{g_0^2 } S_\Lambda(U)\biggr) d\sigma_{\Lambda}(U)\, ,
\]
where $S_\Lambda$ is defined in \eqref{hdef} and $Z(\Lambda, g_0)$ is the normalizing constant, also called the partition function. This probability measure is called the lattice gauge theory on $\Lambda$ for the gauge group $U(N)$, with coupling strength $g_0$. The logarithm of the partition function is called the free energy of the model. Define 
\begin{equation*}
F(\Lambda, g_0) := \frac{\log Z(\Lambda, g_0)}{|\Lambda|}\,,
\end{equation*}
where $|\Lambda|$ denotes the number of vertices in $\Lambda$. This is sometimes called the free energy per site. 

\begin{figure}[t]
\begin{pspicture}(1,.5)(11,5.5)
\psset{xunit=1cm,yunit=1cm}
\psline{*-*}(4.5,1)(7.5,1)
\psline{*-*}(7.5,1)(7.5,4)
\psline{*-*}(7.5,4)(4.5,4)
\psline{*-*}(4.5,4)(4.5,1)
\rput(4.3,.8){$x$}
\rput(8.1, .8){$x+e_j$}
\rput(8.5, 4.2){$x+e_j+e_k$}
\rput(3.9, 4.2){$x+e_k$}
\end{pspicture}
\caption{The plaquette $(x,j,k)$.}
\label{plaquettefig}
\end{figure}
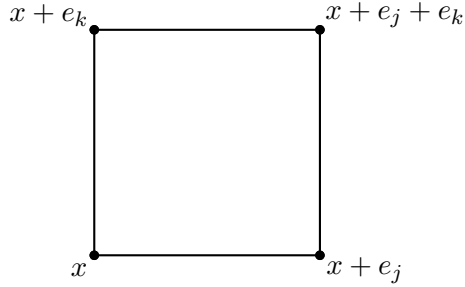

Let $B_{n}$ denote the box $\{0,1,2,\ldots, n-1\}^d$ and let $B_n'$ be the set of plaquettes of $B_n$, as defined above. Consider $U(N)$ lattice gauge theory on $B_n$ with coupling strength $g_0$. Our goal is to give an approximation for $F(B_n, g_0)$ when $n$ is large and $g_0$ is small. This is given in Theorem \ref{mainthm} below. The approximation involves a dimension-dependent constant that takes a few steps to define. We will now define this constant and then state the theorem.

\begin{figure}[t]
\begin{pspicture}(1,.5)(11,7.5)
\psset{xunit=1cm,yunit=1cm}
\psline[linestyle = dashed]{-}(3,1)(9,1)
\psline{-}(3,2)(9,2)
\psline{-}(3,3)(9,3)
\psline{-}(3,4)(9,4)
\psline{-}(3,5)(9,5)
\psline{-}(3,6)(9,6)
\psline{-}(3,7)(9,7)
\psline[linestyle = dashed]{-}(3,1)(3,7)
\psline[linestyle = dashed]{-}(4,1)(4,7)
\psline[linestyle = dashed]{-}(5,1)(5,7)
\psline[linestyle = dashed]{-}(6,1)(6,7)
\psline[linestyle = dashed]{-}(7,1)(7,7)
\psline[linestyle = dashed]{-}(8,1)(8,7)
\psline[linestyle = dashed]{-}(9,1)(9,7)
\end{pspicture}
\caption{The box $B_n$ and its edges, for $n=7$ and $d=2$. The dashed lines represent the edges belonging to $E_n^0$. The solid lines are for edges in $E_n^1$.}
\label{en0fig}
\end{figure}
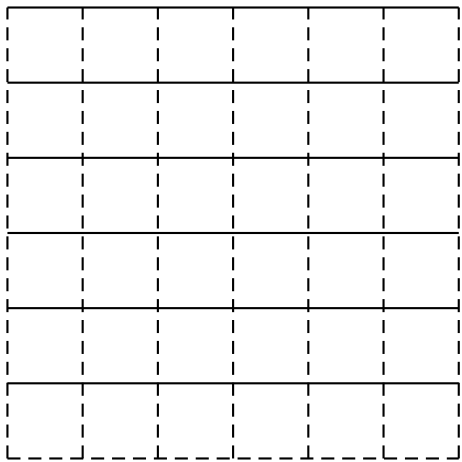

 A nearest-neighbor edge $(x,y)$ of $B_n$ will be called positively oriented if $x$ is smaller than $y$ in the lexicographic ordering. Let $E_n$ be the set of all positively oriented edges of $B_n$. Let $E_n^0$ be the subset of $E_n$ consisting of all edges $(x,y)$ such that for some $1\le j\le d$ and some $x_1,\ldots, x_j$,
\[
x = (x_1,\ldots, x_{j-1}, x_j, 0,0,\ldots,0)
\]
and 
\[
y = (x_1,\ldots, x_{j-1}, x_j+1, 0,0,\ldots, 0)\,.
\]
Let $E_n^1 := E_n \backslash E_n^0$. Expert readers will recognize that the set  $E_n^0$ is related to axial gauge fixing. Pictures of $E_n^0$ and $E_n^1$ in the two-dimensional case are shown in Figure \ref{en0fig}. 

Let $\rr^{E_n}$ be the vector space of all maps from $E_n$ into $\rr$. For $t\in \rr^{E_n}$ and $(x,y)\in E_n$, let $t(y,x):=-t(x,y)$. For any element $t\in \rr^{E_n}$ and any plaquette $(x,j,k)\in B_n'$, let
\begin{align*}
t(x,j,k) &:= t(x,x+e_j) + t(x+ e_j, x+e_j+e_k) \\
&\qquad + t(x+e_j+e_k, x+e_k) + t(x+e_k, x)\,,
\end{align*}
and define
\[
M_n(t) := \sum_{(x,j,k)\in B_n'} t(x,j,k)^2\,.
\]
Then $M_n$ is a quadratic form on $\rr^{E_n}$. Let $M_n^0$ denote the restriction of $M_n$ to the subspace of $\rr^{E_n}$ consisting of all $t$ such that $t(x,y)=0$ for each $(x,y)\in E_n^0$. Since this subspace can be naturally identified with $\rr^{E_n^1}$, $M_n^0$ has a representation as an $|E_n^1| \times |E_n^1|$ matrix. This gives a natural definition of the determinant of $M_n^0$. Define
\begin{equation}\label{kndeq}
K_{n,d} := \frac{-\log \det M_n^0}{2n^d}\,.
\end{equation}
The following theorem gives an asymptotic formula for  $F(B_n,g_0)$ in the limit as $n\ra\infty$ and $g_0\ra 0$, after an appropriate renormalization. This is the main result of this paper. The results about three- and four-dimensional Yang--Mills theories that were mentioned in the abstract and the introduction are consequences of this theorem, stated later in this section.
\begin{thm}\label{mainthm}
Take any $d\ge 2$ and $N\ge 1$. Let $K_{n,d}$ and $F(B_n, g_0)$ be defined as above. Then the limit
\[
K_d := \lim_{n\ra\infty} K_{n,d}
\] 
exists, is finite, and 
\[
\lim_{\substack{n\ra\infty\\ g_0\ra 0}} \biggl(F(B_n, g_0) - \frac{1}{2}\biggl(d-1 -\frac{d}{n}+\frac{1}{n^{d}}\biggr)N^2 \log (g_0^2)\biggr) =  (d-1)\log \biggl( \frac{\prod_{j=1}^{N-1} j!}{(2\pi)^{N/2}}\biggr) + N^2 K_d\,.
\]
\end{thm}
Note that in the above theorem, $n$ and $g_0$ are allowed to vary independently. The terms $d/n$ and $1/n^d$  on the left can be ignored if $\log g_0$ blows up more slowly than $n$. If, on the other hand, $g_0$ goes to zero so fast that $\log  g_0$ blows up faster than $n^d$, then these terms cannot be dropped. The independence of $n$ and $g_0$ in this result does not, however, preclude the possibility that there may be lower order terms that vanish in the limit, whose behaviors are determined by the relation between $n$ and $g_0$. 

Note also that the constant $K_d$ depends only on the dimension $d$. In particular, $K_d$ has no dependence on $N$. It would be interesting to see if $K_d$ can be written explicitly instead of as a limit. 


Very briefly, the intuition behind the proof of Theorem \ref{mainthm} is as follows. When $g_0^2$ is small, the theory behaves like a Gaussian process at small scales. This is the first step of the proof. The second step is to show that the leading term of the free energy is determined wholly by small scale behavior. The main challenge lies in proving that the cumulative effect of the larger scales is of smaller order than the leading term. A more elaborate sketch of the proof is given in Section \ref{sketchsec}.

Let us now consider lattice gauge theory on the scaled lattice $\ep \zz^d$, where we would eventually want to send $\ep$ to zero. The definition of the theory remains exactly the same as for the unscaled lattice, except that $g_0^2$ is replaced by 
\[
g^2 \ep^{4-d}\,,
\]
where $g$ is the new symbol for the coupling strength, and the box $B_n$ is replaced by the scaled box $\ep B_n$. The exponent $4-d$ of the lattice spacing $\ep$ comes from Wilson's prescription for discretizing quantum Yang--Mills theories. The heuristic justification for this exponent will be discussed in Section \ref{litsec}. 

The continuum limit is taken by sending $\ep$ to zero and letting $n$ grow either like $R/\ep$ where $R$ is fixed, which corresponds to the continuum limit in a box of side-length $R$, or letting $n$ grow faster than $1/\ep$, which corresponds to the continuum limit in the whole of $\rr^d$. A key object of interest is the behavior of the free energy of the model under either of the above limits. The following theorem, which is a simple corollary of Theorem \ref{mainthm}, gives a formula for the leading term of the free energy in the continuum limit in dimensions two and three.
\begin{thm}\label{d23thm}
Take any $N\ge 1$ and take $d= 2$ or $3$. Let $Z(n,\ep,g)$ be the partition function of $d$-dimensional $U(N)$ lattice gauge theory with coupling strength $g$ in the box $\{0,\ep,2\ep,\ldots, (n-1)\ep\}^d$, as defined above. Suppose that $n\ra\infty$ and $\ep\ra 0$ simultaneously such that $n\ep$ stays bounded away from zero. Then for any $g>0$,
\begin{align*}
\log Z(n, \ep, g) &= n^d \biggl( \frac{1}{2}(d-1)N^2\log (g^2\ep^{4-d}) + (d-1)\log \biggl( \frac{\prod_{j=1}^{N-1} j!}{(2\pi)^{N/2}}\biggr) + N^2 K_d\biggr) + o(n^d)\,,
\end{align*}
where $K_d$ is as in Theorem \ref{mainthm}, and $o(n^d)$ denotes a term that tends to zero when divided by $n^d$.
\end{thm}

In dimension four, the notable feature is that  the coupling strength $g^2 \ep^{4-d}$ does not depend on~$\ep$, since $4-d=0$.  As a consequence, the partition function $Z(n,\ep,g)$ has no explicit dependence on~$\ep$. (It seems, however, that nowadays most physicists believe that there should some logarithmic correction.) Theorem~\ref{mainthm} still gives some useful information, but only in the weak coupling limit --- that is, if we send $g$ to zero simultaneously as $n$ is sent to infinity and $\ep$ is sent to zero.
\begin{thm}\label{d4thm}
Take any $N\ge 1$ and take $d=4$. Let $Z(n,\ep,g)$ be as in Theorem~\ref{d23thm} and $K_d$ be as in Theorem~\ref{mainthm}. Suppose that $n\ra\infty$, $\ep\ra 0$ and $g\ra 0$ simultaneously, with no restrictions on how they are related to each other. Then 
\begin{align*}
\log Z(n,\ep, g) &= n^4\biggl(\biggl(\frac{3}{2} -\frac{2}{n}+\frac{1}{2n^{4}}\biggr)N^2\log (g^2) + 3\log \biggl( \frac{\prod_{j=1}^{N-1} j!}{(2\pi)^{N/2}}\biggr) + N^2 K_4\biggr) + o(n^4)\,.
\end{align*}
\end{thm}
Since it is clearly evident how Theorem \ref{d23thm} and Theorem \ref{d4thm} follow from Theorem \ref{mainthm}, we will only prove Theorem \ref{mainthm}.  Section \ref{litsec} contains a review of the physical and mathematical backgrounds and a discussion of existing results.  A sketch of the proof is given in Section \ref{sketchsec}. Section \ref{notations} contains a list of notations and conventions that will be used in the proof. The remaining sections are devoted to the proof of Theorem \ref{mainthm}. These sections contain a number of results that may be of independent interest.

 
\section{Background and literature}\label{litsec}
Recall that the Lie algebra $\mathfrak{u}(N)$ of the Lie group $U(N)$ is the set of all $N\times N$ skew-Hermitian matrices. A $U(N)$ connection form on $\rr^d$ is a smooth map from $\rr^d$ into $\mathfrak{u}(N)^d$. If $A$ is a $U(N)$ connection form, its value $A(x)$ at a point $x$ is a $d$-tuple $(A_1(x),\ldots, A_d(x))$ of skew-Hermitian matrices. In the language of differential forms, 
\[
A = \sum_{j=1}^d A_j dx_j\,.
\]
The curvature form $F$ of a connection form $A$ is the $\mathfrak{u}(N)$-valued $2$-form
\[
F = dA + A \wedge A\,.
\]
This means that at each point $x$, $F(x)$ is a $d\times d$ array of skew-Hermitian matrices of order $N$, whose $(j,k)^{\mathrm{th}}$ entry is the matrix
\[
F_{jk}(x) = \fpar{A_k}{x_j} - \fpar{A_j}{x_k} + [A_j(x), A_k(x)]\,,
\]
where $[B, C] = BC - CB$ denotes the commutator of two matrices $B$ and $C$.

Let $\ma$ be the space of all $U(N)$ connection forms on $\rr^d$. The Yang--Mills action on this space is the function
\[
S_{\mathrm{YM}}(A) := -\int_{\rr^d} \tr(F\wedge *F) \,,
\]
where $F$ is the curvature form of $A$ and $*$ denotes the Hodge $*$-operator. Explicitly, this is
\[
S_{\mathrm{YM}}(A) = -\int_{\rr^d}\sum_{j,k=1}^d \tr(F_{jk}(x)^2)\,dx \,.
\]
The Euclidean version of $U(N)$ quantum Yang--Mills theory on $\rr^d$ (henceforth shortened as ``Euclidean Yang--Mills theory'') is informally described as the probability measure
\[
d\mu(A) = \frac{1}{Z}\exp\biggl(-\frac{1}{4g^2}S_{\mathrm{YM}}(A)\biggr) \dd A\,,
\]
where $A$ belongs to the space $\ma$ of all $U(N)$ connection forms, $S_{\mathrm{YM}}$ is the Yang--Mills functional defined above,  
\[
\dd A = \prod_{j=1}^d \prod_{x\in \rr^d} d(A_j(x))
\]
is ``infinite-dimensional Lebesgue measure'' on $\ma$, $g$ a positive coupling constant, and $Z$ is the normalizing constant that makes this a probability measure. 

The above description of Euclidean Yang--Mills theory is not mathematically valid, partly due to the non-existence of an infinite-dimensional Lebesgue measure on $\ma$. The problem of giving a meaning to Euclidean Yang--Mills theory, in a way that makes it possible to extend the definition to the Minkowski version of quantum Yang--Mills theory (where the additional feature is that there is an $\I$ in the exponent) via Wick rotation, is one of the central open questions of mathematical physics. The importance of this question was emphasized by its inclusion in the list of ``millennium prize problems'' posed by the Clay Institute (see \cite{jaffewitten}). The reason behind the importance of this question is that quantum Yang--Mills theories are the building blocks of the Standard Model of quantum mechanics.

The most popular approach to constructing Euclidean Yang--Mills theories was proposed by \cite{wilson74}. Wilson's idea was to define a discrete version of the theory, with the hope of taking a continuum limit as the lattice spacing is sent to zero. This is the lattice gauge theory that was defined in Section \ref{results}. Although we defined it only for the group $U(N)$, it can be similarly defined for other Lie groups. The group $SU(3)\times SU(2)\times U(1)$ is the one that arises in the Standard Model.

Wilson's logic behind the definition of lattice gauge theory goes as follows. For simplicity, we will restrict the discussion to $U(N)$ lattice gauge theory. First, let us discretize the space $\rr^d$ as the scaled lattice $\ep \zz^d$. Take a connection form $A\in \ma$ and a point $x\in \ep\zz^d$. Then $A(x)$ is made up of $d$ components $A_1(x),\ldots, A_d(x)$, each of which is an element of $\mathfrak{u}(N)$. For each $j$, Wilson obtains from $A_j(x)$ a unitary matrix $U_j(x) = e^{\ep A_j(x)}$. The matrix $U_j(x)$ is now rewritten as a matrix $U(x,x+\ep e_j)$ attached to the positively oriented edge $(x,x+\ep e_j)$, with the usual convention that $U(y,x)=U(x,y)^{-1}$. By this prescription, we obtain a configuration of unitary matrices attached to edges from the values of a connection form on the vertices of $\ep\zz^d$. Formally, the Wilson action for a configuration $U$ is given by
\[
S(U) = \sum_{x\in \ep\zz^d} \sum_{1\le j<k\le d}\phi(U(x,j,k))\,,
\]
where $\phi$ and $U(x,j,k)$ are as in Section \ref{results}. Recall that
\begin{align*}
U(x,j,k) &= U(x,x+\ep e_j)U(x+\ep e_j, x+\ep e_j+\ep e_k) U(x+\ep e_j+\ep e_k, x+\ep e_k) U(x+\ep e_k,x)\\
&= e^{\ep A_j(x)} e^{\ep A_k(x+\ep e_j)} e^{-\ep A_j(x+\ep e_k)} e^{-\ep A_k(x)}\,.
\end{align*}
Recall the Baker--Campbell--Hausdorff formula for products of matrix exponentials:
\begin{align*}
e^B e^C &= \exp\biggl(B + C + \frac{1}{2}[B,C] + \text{higher commutators}\biggr)\,.
\end{align*}
Iterating this gives
\begin{align*}
e^{B_1}e^{B_2}\cdots e^{B_n} &= \exp\biggl(\sum_{j=1}^n B_j + \frac{1}{2}\sum_{1\le j<k\le n}[B_j,B_k] + \text{higher commutators}\biggr)\,.
\end{align*}
Next, recall that the eigenvalues of a skew-Hermitian matrix are all purely imaginary, and that the commutator of two skew-Hermitian matrices is skew-Hermitian. Consequently, the term within the exponential on the right side of the above display is skew-Hermitian and therefore has a purely imaginary trace. This implies that if the entries of the matrices $B_1,\ldots,B_n$ are of order $\ep$ and if the entries of $B_1+\cdots+B_n$ are of order $\ep^2$, then
\begin{align*}
\Re(\tr(I- e^{B_1}e^{B_2}\cdots e^{B_n})) &= -\frac{1}{2}\tr\biggl[\biggl(\sum_{j=1}^n B_j + \frac{1}{2}\sum_{1\le j<k\le n}[B_j,B_k]\biggr)^2\biggr] + O(\ep^5)\,,
\end{align*}
where the real part of the trace was replaced by the trace on the right because the square of a skew-Hermitian matrix has real eigenvalues. Writing
\[
A_k(x+\ep e_j) = A_k(x)+\ep \fpar{A_k}{x_j} + O(\ep^2)\,, \ \ \ A_j(x+\ep e_k) = A_j(x) + \ep \fpar{A_j}{x_k} + O(\ep^2)\,,
\]
we see that 
\begin{align*}
A_j(x)+ A_k(x+\ep e_j)- A_j(x+\ep e_k) -A_k(x) = \ep\biggl(\fpar{A_k}{x_j} - \fpar{A_j}{x_k}\biggr) + O(\ep^2)\,.
\end{align*}
Therefore, by the remarks made above,
\begin{align*}
\phi(U(x,j,k)) &= \Re(\tr(I-e^{\ep A_j(x)} e^{\ep A_k(x+\ep e_j)} e^{-\ep A_j(x+\ep e_k)} e^{-\ep A_k(x)}))\\
&= -\frac{1}{2}\ep^4 \tr\biggl[\biggl(\fpar{A_k}{x_j} - \fpar{A_j}{x_k} + [A_j(x), A_k(x)]\biggr)^2\biggr] + O(\ep^5)\\
&= -\frac{1}{2}\ep^4\tr(F_{jk}(x)^2)+O(\ep^5)\,.
\end{align*}
This gives the formal approximation
\begin{align*}
-\frac{1}{g^2\ep^{4-d}}S(U)  &= -\frac{1}{g^2\ep^{4-d}}\sum_{x\in \ep\zz^d}\sum_{1\le j<k\le d} \phi(U(x,j,k))\\
&\approx \frac{1}{4g^2} \sum_{x\in \ep \zz^d}\sum_{j,k=1}^d\ep^d \tr(F_{jk}(x)^2)\\
&\approx \frac{1}{4g^2}\int_{\rr^d}\sum_{j,k=1}^d \tr(F_{jk}(x)^2)\, dx = -\frac{1}{4g^2}S_{\mathrm{YM}} (A)\,.
\end{align*}
The above argument was used by Wilson to justify the approximation of Euclidean Yang--Mills theory by lattice gauge theory, with the $\ep^{4-d}$ term in front of the action. Although the heuristic looks fairly convincing, no one has been able to make rigorous mathematical sense of the convergence of lattice gauge theory to its continuum limit, except in some special cases. The reason for the specific form of the Wilson action, instead of something simpler, is that any discretization of the Yang--Mills action must retain a crucial property known as gauge invariance. The Wilson action has this property, but is not the only one. An action based on the heat kernel on the gauge group, known as the Villain action, is also gauge invariant and quite popular.

The case in which mathematicians have been the most successful in taking the continuum limit is two-dimensional lattice gauge theory. There is now a nearly complete body of work on this topic. The two-dimensional Higgs model, which is $U(1)$ Yang--Mills theory with an additional Higgs field, was constructed by \cite{bfs79, bfs80, bfs81} and further refined by \cite{bs83}. Building on an idea of \cite{bralic80}, \cite{gks89} formulated a rigorous mathematical approach to performing calculations in two-dimensional Yang--Mills theories via stochastic calculus. Somewhat different ideas leading to the same goal were implemented by \cite{driver89a, driver89b} and \cite{kk87}. The papers of \cite{driver89a, driver89b} made precise the idea of using objects called lassos to define the continuum limit of Yang--Mills theories. Explicit formulas for Yang--Mills theories on compact surfaces were obtained by \cite{fine90, fine91} and \cite{witten91, witten92}. All of these results were generalized and unified by \cite{sengupta92, sengupta93, sengupta97} using the stochastic calculus approach. 

More recently, \cite{levy03, levy10} has introduced an abstract framework for constructing  two-dimensional Yang--Mills theories as random holonomy fields. A random holonomy field is  a stochastic process indexed by curves on a surface, subject to boundary conditions, and behaving under surgery as dictated by a Markov property. L\'evy's framework allows parallel transport along more general curves than the ones considered previously, and makes interesting connections to topological quantum field theory. A relatively non-technical description of this body of work is given in the survey of \cite{levy11}. For some very recent developments in rigorously verifying the validity of theoretical physics results for two-dimensional Yang--Mills theories, see \cite{nguyen15}. 

Euclidean Yang--Mills theories in  dimensions three and four have proved to be more challenging to construct mathematically. At sufficiently strong coupling, a number of conjectures about lattice gauge theories --- such as quark confinement and the existence of a positive self-adjoint transfer matrix --- were rigorously proved by \cite{os78}. The strong coupling techniques, unfortunately, do not help in constructing the continuum limit. The most popular approach to showing the existence of a continuum limit of lattice gauge theories is the so-called phase cell renormalization technique.  Phase cell renormalization can be briefly described as follows. Consider $U(N)$ lattice gauge theory on the scaled lattice $\ep\zz^d$. Let $L$ be a fixed positive integer and suppose that $\ep = L^{-k}$ for some $k$.  Let $S$ be the Wilson action, and let
\[
S_{k,k}(U) := \frac{1}{g^2\ep^{4-d}} S(U)\,.
\]
The first step in phase cell renormalization is to generate, given a configuration $U_k$ of unitary matrices attached to edges of $\ep\zz^d$, a configuration $U_{k-1}$ of unitary matrices attached to edges of the coarser lattice $L\ep \zz^d$. A survey of the various ways of producing $U_{k-1}$ from $U_k$ is given in Chapter 22 of the classic monograph of \cite{gj87}. The next step is to understand the probability density function of the configuration $U_{k-1}$. The probability density is written as a function proportional to $$e^{-S_{k,k-1}(U_{k-1})}\,,$$ 
where $S_{k,k-1}$ is called the effective action. Computing the effective action explicitly is usually impossible, but qualitative features and bounds are sometimes possible to obtain.

The final step is to iterate this process. From the configuration $U_{k-1}$ on $L\ep\zz^d$ one obtains a configuration $U_{k-2}$ on $L^2\ep\zz^d$, from $U_{k-2}$ a configuration $U_{k-3}$ on $L^3\ep\zz^d$, and so on. The process is stopped at stage $k$. The effective action at stage $j$ is denoted by $S_{k,k-j}$. The effective action $S_{k,0}$ describes the behavior of parallel transport along macroscopic curves. Presumably, $S_{k,0}$ depends on $k$, but it may so happen that it has a limit as $k\ra\infty$. The goal of the renormalization process is to show that $S_{k,0}$ has a limit as $k\ra \infty$ and to understand this limit. If this objective turns out to be too difficult, which is usually the case, one would at least like to show that the sequence $\{S_{k,0}\}_{k\ge 1}$ can be embedded in a space of functions that is compact in some suitable topology. This would demonstrate the existence of subsequential limits of the macroscopic theory as the lattice spacing $\ep$ is sent to zero. 

The approximation of continuum Yang--Mills theory by a lattice gauge theory on a lattice with spacing $\ep$ is analogous to truncating a Fourier series at a finite frequency. For this reason mathematical physicists often refer to lattice approximations as ultraviolet cutoff (ultraviolet = high frequency). For the same reason, the compactness of $\{S_{k,0}\}_{k\ge 1}$ posited in the previous paragraph is called ultraviolet stability of the effective action, that is, stability with respect to the cutoff frequency. 

A notable success story of phase cell renormalization is the work of \cite{king86a, king86b}, who established the existence of the continuum limit of the three-dimensional Higgs model. The continuum limit of pure $U(1)$ Yang--Mills theory (that is, without the Higgs field) was established earlier by \cite{gross83}, but with a different notion of convergence. 

Ultraviolet stability of three- and four-dimensional non-Abelian lattice gauge theories by phase cell renormalization, as outlined above, was famously established by \cite{balaban83, balaban84a, balaban84b, balaban84c, balaban84d, balaban85a, balaban85b, balaban85c, balaban85d, balaban85e, balaban87, balaban88, balaban89a, balaban89b} in a long series of papers spanning six years.  A somewhat different approach, again using phase cell renormalization, was pursued by \cite{federbush86, federbush87a, federbush87b, federbush88, federbush90} and \cite{fw87}.  

A completely different idea was worked out by \cite{mrs93}, who started with a Lie algebra action instead of a Lie group action. They broke the gauge invariance of this action by augmenting it with an extra non-gauge invariant quadratic part chosen so that this quadratic part  defined a normalizable Gaussian measure.  Furthermore this Gaussian measure is supported on smooth connections such that the nonlinear terms in the action made sense.   In this theory, which is defined in the continuum, this extra term served the role of an ultraviolet cutoff instead of the lattice.  The extra term must subsequently be removed by taking a limit outside the integral.  The main contribution of \cite{mrs93} was to formulate a program to prove that gauge invariance is recovered in this limit so rapidly that renormalization with finitely many counterterms is still possible.  

In spite of the remarkable achievements surveyed above, the progress on the important question of constructing Euclidean Yang--Mills theories in the continuum (and indeed, the progress of constructive quantum field theory as a whole) has stalled, partly due to the daunting complexity of the renormalization methods employed in the most advanced papers and partly due to the unavailability of gauge-invariant observables that remain well-behaved in the continuum limit in dimensions three and four. In the absence of well-behaved observables, one can try to understand the behavior of the partition function.  Understanding the partition function is only one aspect of constructing Euclidean Yang--Mills theories, but it is an important aspect.  The importance of understanding the Yang--Mills partition function has been highlighted, for example, by \cite{douglas04} in a status report for the millennium prize problem.  A complete solution of the problem would involve understanding the exact asymptotics of the partition function rather than just the leading term in the exponent that is derived in this manuscript. Understanding the exact asymptotics is necessary for passing from Euclidean spacetime to Minkowski spacetime by analytic continuation.

In the context of the previous paragraph, it is important to note one recent development. \cite{cg13, cg15} have proposed a new approach to the problem of construction of gauge invariant observables that are well-behaved in the continuum limit. The method proposed in these papers is based on regularizing a connection form by the Yang--Mills heat equation.

Incidentally, besides the problem of constructing the continuum limit, there are a number of other important mathematical open problems related to lattice gauge theories. The mass gap problem, quark confinement and the $1/N$ expansion are three examples. For a survey of some of these problems and some recent progress on the $1/N$ expansion, see \cite{chatterjee15}.  A method for non-perturbative constructions of Yang--Mills theories that has received considerable attention in the physics literature in recent years goes by the name of ``resurgence and transseries in QFT''. For surveys of this line of work, see \cite{uy08} and \cite{du16}. For a rigorous mathematical approach to perturbative constructions of quantum Yang--Mills theories and other quantum field theories, see \cite{costello11}. 


\section{The approach of this paper}\label{sketchsec}
The key idea behind the proof of Theorem \ref{mainthm} is that when $g_0$ is small, $U(N)$ lattice gauge theory behaves like a Gaussian theory at small scales, and this behavior determines the leading term of the free energy. The proof is executed by first computing an upper bound for the free energy, and then establishing a matching lower bound. 

The main observation that leads to the upper bound is that $\phi$ is a nonnegative function, which implies that if we drop some plaquettes from the Hamiltonian, the partition function increases. On the other hand, if we represent a large box $B_n$ as a union of non-overlapping small boxes of width $m$, and remove the plaquettes that touch the boundaries of multiple boxes, then the resulting integral breaks up as a product of integrals. Using these observations, we can establish that if $n$ is a multiple of $m$, then $F(B_n, g_0)\le F(B_m, g_0)$. More generally, for any $n\ge m$, $F(B_n, g_0)\le (1-Cm/n)F(B_m, g_0)$ where $C$ depends only on $N$ and $d$. This allows us to reduce the upper bound problem to understanding the behavior of $F(B_m,g_0)$ as $m$ remains fixed (or grows very slowly) and~$g_0\ra 0$. 

If $g_0$ is small, then one can argue that with high probability, $\phi(U(x,j,k))$ must be small for a typical plaquette. This would imply that $U(x,j,k)$ is close to the identity matrix for a typical plaquette. The closeness of $U(x,j,k)$ to $I$ implies that any three of the four unitary matrices attached to the edges of the plaquette approximately determine the fourth. If $g_0$ is small and $m$ is not too large, then the above property can be used to deduce inductively that the matrices attached to the edges in $E_m^0$ approximately determine the matrices attached to all the edges in $E_m$. This allows us to carry out the integration for computing $F(B_m,g_0)$ by first fixing the values of the matrices attached to the edges in $E_m^0$, then using a kind of Laplace approximation --- taking advantage of the fact that the configuration space is now reduced to a small neighborhood of a particular configuration --- and finally noting that the integral does not actually depend on the values of the matrices attached to $E_m^0$ (by axial gauge fixing). Note that the whole argument hinged on the fact that we could let $m$ grow very slowly as $g_0\ra 0$. Otherwise, the matrices on $E_m^0$ do not approximately determine the other matrices. 

For the lower bound, however, we cannot directly replace a large box by a small box, because the inequality goes in the opposite direction. The technique for the lower bound can be roughly described as follows. Take some large $n$ and small $g_0$. If all the $U(x,j,k)$'s were exactly equal to $I$, then the matrices on $E_n^0$ would exactly determine all the other matrices. Since $g_0$ is small, we can deduce as before that a typical $U(x,j,k)$ is close to $I$. However, since $n$ can be arbitrarily large, we can no longer deduce from this that the matrices on $E_n^0$ approximately determine the matrices on all the other edges. 

The main idea that helps us cross this hurdle is that we can {\it pretend} that the matrices on $E_n^0$ approximately determine the matrices on the other edges. That is, given a set of matrices attached to edges in $E_n^0$, let $A$ be the set of all configurations where the matrices attached to the edges in $E_n$ are all close to the matrices that we would get if the $U(x,j,k)$'s were all exactly equal to $I$. Then 
\[
Z(B_n, g_0) = \frac{1}{p}\int_A e^{- S_{B_n}(U)/g_0^2}d\sigma_{B_n}(U)\,,
\]
where
\[
p = \text{the probability of $A$ under the lattice gauge model.}
\]
The integral on the right can be evaluated, as before, by a kind of Laplace approximation. The above representation shows that this integral can be used as a surrogate for $Z(B_n,g_0)$ if $p = \exp(o(n^d))$, since $\log Z(B_n, g_0)$ is of order $n^d$. Now, $p$ is the probability that roughly $n^d$ matrices are close to some prescribed values. It is therefore natural to expect that $p$ behaves like $\exp(Cn^d)$ for some constant $C$ rather than like $\exp(o(n^d))$. The reason why we can show that $p$ behaves like $\exp(o(n^d))$ is that most of the $U(x,j,k)$'s are close to $I$, which allows us to control the set of all matrices by controlling a subset of size $o(n^d)$. For example, if we represent the box $B_n$ as a union of boxes of width $m$, where $m$ is chosen to be large but not too large (depending on $g_0$), then the matrices on the edges that lie on the boundaries of these small boxes approximately determine all the other matrices, and there are $o(n^d)$ of these boundary edges. This allows us to show that $p = \exp(o(n^d))$ and complete the proof of the lower bound.

\section{Notations and conventions}\label{notations}
Throughout the remainder of this article, the dimension $d$ and the order $N$ of the unitary group $U(N)$ will be treated as fixed constants. We will write $C, C_0, C_1,\ldots$ for positive constants that depend only on $N$ and $d$. In some sections, where $N$ is not involved, these constants will only depend on $d$. The values of these constants may change from line to line. 

For notational convenience, we will use the variable
\[
\beta := \frac{1}{g_0^2}
\] 
instead of $g_0$ in our arguments. Sometimes, $\beta$ and $g_0$ may appear within the same line, especially in the statements of theorems and lemmas.

In addition to the above, we will be using a large number of symbols and notations. Many of these have already been defined in Section \ref{results}, and the rest will be defined in the forthcoming sections. For easy reference, we summarize these with short descriptions in the following table.
\vskip.2in
\begin{longtable}{c c p{.5\textwidth}}
Notation & Defined in section: & Short description\\
\hline
$d$ & \ref{results} & Dimension of the lattice $\zz^d$.\\
$e_1,\ldots, e_d$ & \ref{results} & Standard basis of $\rr^d$.\\
$N$ & \ref{results} & Order of unitary group $U(N)$.\\
$U(N)$ & \ref{results} & Unitary group of order $N$.\\
$\sigma$ & \ref{results} & Haar measure on $U(N)$.\\
$B_n$ & \ref{results} & The box $\{0,1,\ldots, n-1\}^d$.\\
$E_n$ & \ref{results}  & Positively oriented edges of $B_n$.\\
$E_n^0$ & \ref{results}  & A subset of $E_n$, related to axial gauge fixing.\\
$E_n^1$ & \ref{results} & $E_n\backslash E_n^0$.\\
$|\Lambda|$ & \ref{results} & The size of a set $\Lambda$.\\
Plaquette & \ref{results} & A square in $\zz^d$, identified by a triple consisting of a vertex and two coordinate directions.\\
$\Lambda'$ & \ref{results} & The set of all plaquettes with all vertices belonging to a set $\Lambda \subseteq \zz^d$.\\
$U(B_n)$ & \ref{results} & Sets of configurations of unitary matrices attached to edges of $B_n$.\\
$\sigma_{B_n}$ & \ref{results}  & Product Haar measure on $U(B_n)$.\\
$\phi$ & \ref{results}  & $\phi(U)=\Re(\tr(I-U))$.\\
$S_{B_n}$ & \ref{results}  & Action of lattice gauge theory on $U(B_n)$.\\
$g_0$ & \ref{results} & Coupling strength in lattice gauge theory.\\
$\mu_{B_n, g_0}$ & \ref{results} & The probability measure on $U(B_n)$ defined by lattice gauge theory.\\
$Z(B_n, g_0)$ & \ref{results} & Partition function of lattice gauge theory on $B_n$.\\
$F(B_n, g_0)$ & \ref{results} & Free energy per site of lattice gauge theory on $B_n$.\\
$M_n$ & \ref{results} & Quadratic form on $\rr^{E_n}$. Lattice Maxwell action.\\
$M_n^0$ & \ref{results} & Restriction of $M_n$ to a subspace of $\rr^{E_n}$.\\
$K_{n,d}$ & \ref{results} & Constant defined in Section \ref{results}, related to lattice Maxwell theory. \\
$K_d$ & \ref{results} & Limit of $K_{n,d}$ as $n\ra\infty$.\\
$\beta$ & \ref{notations} & $1/g_0^2$.\\
$\|M\|$ & \ref{smallsec} & Hilbert--Schmidt norm of a complex matrix $M$.\\
GUE matrix & \ref{smallsec} & A random matrix from the Gaussian Unitary Ensemble.\\
$U_0(B_n)$ & \ref{gaugesec} & Set of all $U\in U(B_n)$ such that $U(x,y)=I$ for $(x,y)\in E_n^0$.\\
$\sigma^0_{B_n}$ & \ref{gaugesec} & Product Haar measure on $U_0(B_n)$.\\
$G(B_n)$ & \ref{gaugesec} & Set of gauge transforms of $U(B_n)$.\\
$U_0^\beta(B_n)$ & \ref{uppersec} & A subset of $U_0(B_n)$ defined in Theorem \ref{mainstep1}.\\
$|x|_1$ & \ref{uppersec} & $\ell^1$ norm of $x$.\\
$H(N)$ & \ref{liesec} & The set of all $N\times N$ complex Hermitian matrices.\\
$B(U,r)$ & \ref{liesec} & The set of all $V\in U(N)$ such that $\|U-V\|\le r$.\\
$b(H, r)$ & \ref{liesec} & The set of all $G\in H(N)$ such that $\|H-G\|\le r$.\\
$\lambda$ & \ref{liesec} & Lebesgue measure on $H(N)$.\\
$C_N$ & \ref{liesec} & A constant defined in the statement of Theorem~\ref{liethm}.\\
$T(B_n)$ & \ref{maxwellsec} & $\rr^{E_n}$.\\ 
$T_0(B_n)$ & \ref{maxwellsec} & The set of all $t\in T(B_n)$ that are zero on $E_n^0$.\\
$T_E(B_n)$ & \ref{maxwellsec} & The set of all $t\in T(B_n)$ that are zero on some set $E\subseteq E_n$.\\
$M_{E,\theta, n}$ & \ref{maxwellsec} & A quadratic form on $T_E(B_n)$.\\
$\tau_{E,\theta, n}$ & \ref{maxwellsec} & The Gaussian measure defined by the quadratic form $M_{E,\theta, n}$.\\
$\tau_n$ &\ref{maxwellsec} & $\tau_{E,\theta, n}$ with $E=E_n^0$ and $\theta=0$.\\ 
$l^0_n$ & \ref{maxwelllim} & Lebesgue measure on $T_0(B_n)$.\\
$Z_M(B_n)$ & \ref{maxwelllim} & Partition function of lattice Maxwell theory.\\
$F_M(B_n)$ & \ref{maxwelllim} & $n^{-d}\log Z_M(B_n)$.\\
$E_n'$ & \ref{maxwelllim} & The union of $E_n^0$ and the set of all boundary edges of $B_n$.\\
$T'(B_n)$ & \ref{maxwelllim} & $T_{E_n'}(B_n)$.\\
$l_n'$ & \ref{maxwelllim} & Lebesgue measure on $T'(B_n)$.\\
$Z_M'(B_n)$ & \ref{maxwelllim} & A small perturbation of $Z_M(B_n)$ that is used in the proof of Theorem~\ref{cdthm}.\\
$F_M'(B_n)$ & \ref{maxwelllim} & $n^{-d}\log Z_M'(B_n)$.\\
$Z_{M,m}(B_n)$ & \ref{maxwelllim} & Another perturbation of $Z_{M}(B_n)$ used in the proof of Theorem \ref{cdthm}.\\
$F_{M,m}(B_n)$ & \ref{maxwelllim} & $n^{-d}\log Z_{M,m}(B_n)$.\\
$E_n^m$ & \ref{maxwelllim} & A subset of $E_n$ appearing in the proof of Theorem~\ref{cdthm}.\\
$T_m(B_n)$ & \ref{maxwelllim} & $T_{E_n^m}(B_n)$.\\
$l^m_n$ & \ref{maxwelllim} & Lebesgue measure on $T_m(B_n)$.\\
$H(B_n)$ & \ref{wmsec} & Sets of configurations of Hermitian matrices attached to edges of $B_n$.\\
$H_0(B_n)$ & \ref{wmsec} & Set of all $H\in H(B_n)$ that are zero on $E_n^0$.\\
$M_n(H)$ & \ref{wmsec} & Lattice Maxwell action of a configuration $H\in H(B_n)$.\\
\hline
\end{longtable}
\vskip.2in

\section{Small ball probabilities for the unitary group}\label{smallsec}
Recall that the Hilbert--Schmidt norm of an $N\times N$ complex matrix $M = (m_{ij})_{1\le i,j\le N}$  is defined~as
\[
\|M\| := \biggl(\sum_{i,j=1}^N |m_{ij}|^2\biggr)^{1/2} = (\tr(M^*M))^{1/2}\,.
\]
It is easy to check that this satisfies the usual properties of a norm. Moreover, it also satisfies the defining property of a matrix norm, that is, 
\[
\|M_1M_2\|\le \|M_1\|\|M_2\|\,.
\]
This is verified easily by an application of the Cauchy--Schwarz inequality.

Since $U(N)$ is a subset of the space of all $N\times N$ complex matrices, the Hilbert--Schmidt norm induces a metric on $U(N)$. The topology induced by this metric is compact and the left and right multiplication operations are continuous in this topology. Therefore we can consider the Haar measure on $U(N)$ as a measure defined on the Borel sigma-algebra of this topology. Throughout this paper, a Borel measurable function on $U(N)$ will mean a function that is measurable with respect to this sigma-algebra.

The goal of this section is to prove the following theorem, which gives an asymptotic approximation for the Haar measure of small balls of the Hilbert--Schmidt metric on~$U(N)$. 
\begin{thm}\label{smallballthm}
Let $\sigma$ denote the Haar measure on $U(N)$. Then
\[
\lim_{\delta \ra 0}\frac{\sigma(\{U: \|I-U\|\le \delta\})}{\delta^{N^2}} = \frac{\prod_{j=1}^{N-1} j!}{(2\pi)^{N/2} 2^{N^2/2}\Gamma(N^2/2+1)}\,.
\]
\end{thm}
The proof is done in two steps. The first step is the following lemma. The proof of this lemma is based on the joint density of eigenvalues of a random unitary matrix, the formula for which was derived by \cite{weyl39}. For modern discussions of Weyl's formula and its connections to the famous integral formula of \cite{selberg}, see \cite{ds94} and \cite{fw08}.
\begin{lmm}\label{matrixlmm1}
Let $B(r)$ denote the Euclidean ball of radius $r$ around the origin in $\rr^N$. Then 
\[
\lim_{\delta \ra 0}\frac{\sigma(\{U: \|I-U\|\le \delta\})}{\delta^{N^2}} = \frac{1}{(2\pi)^NN!} \int_{B(1)}  \prod_{1\le j<k\le N} (\eta_j-\eta_k)^2\, d\eta_1\,d\eta_2\,\cdots \,d\eta_N\,.
\]
\end{lmm}
\begin{proof}
Suppose that $h:U(N) \ra \cc$ is a function having the property that $h(U)=f(\lambda_1,\ldots,\lambda_N)$, where  $\lambda_1,\ldots,\lambda_N$ are the eigenvalues of $U$ repeated by multiplicities, and  $f$ is some measurable function of $N$ complex variables that is symmetric under permutation of coordinates. In other words, $h(U)$ is a symmetric function of the eigenvalues of $U$. It was proved by \cite{weyl39} that under these circumstances,
\begin{align*}
\int_{U(N)} h(U)\, d\sigma(U) &= \frac{1}{(2\pi)^N N!} \int_{[-\pi,\pi]^N} f(e^{\I\theta_1},\ldots, e^{\I\theta_N}) \prod_{1\le j<k\le N} |e^{\I\theta_j} - e^{\I\theta_k}|^2\, d\theta_1\,d\theta_2\,\cdots \,d\theta_N\,.
\end{align*}
Now let
\[
h(U) := \|I-U\|^2\,.
\]
The eigenvalues of $U^*$ are the complex conjugates of the eigenvalues of $U$, and $U^*U=I$. Therefore, if $\lambda_1,\ldots, \lambda_N$ are the eigenvalues of $U$, then
\begin{align}
h(U) &= \tr((I-U)^*(I-U))\nonumber\\
&= \tr(2I-U^*-U)\nonumber\\
&= \sum_{j=1}^N(2 - \bar{\lambda}_j-\lambda_j)= 2\sum_{j=1}^N (1-\Re(\lambda_j))\,.\label{traceform}
\end{align}
Thus, for any smooth $\psi:[0,\infty) \ra \rr$ with compact support,
\begin{align*}
&\int_{U(N)} \psi(h(U)/\delta^2)\,d\sigma(U)  \\
&= \frac{1}{(2\pi)^NN!} \int_{[-\pi,\pi]^N} \psi\biggl(\frac{2}{\delta^{2}} \sum_{j=1}^N (1-\cos\theta_j)\biggr) \prod_{1\le j<k\le N} |e^{\I\theta_j} - e^{\I\theta_k}|^2\, d\theta_1\,d\theta_2\,\cdots\, d\theta_N\,.
\end{align*}
Making the change of variable $\eta_i = \theta_i/\delta$, this gives
\begin{align*}
&\int_{U(N)} \psi(h(U)/\delta^2)\,d\sigma(U)  \\ &= \frac{\delta^{N^2}}{(2\pi)^NN!} \int_{[-\pi/\delta,\pi/\delta]^N} \psi\biggl(2\sum_{j=1}^N \frac{1-\cos\delta\eta_j}{\delta^2}\biggr) \prod_{1\le j<k\le N} \biggl|\frac{e^{\I\delta\eta_j} - e^{\I\delta\eta_k}}{\delta}\biggr|^2\, d\eta_1\,d\eta_2\,\cdots \,d\eta_N\,.
\end{align*}
Dividing both sides by $\delta^{N^2}$, sending $\delta\ra 0$, and using the compactness of the support of $\psi$ and the smoothness of $\psi$, we get
\begin{align}\label{matrix1}
&\lim_{\delta \ra 0} \delta^{-N^2}\int_{U(N)} \psi(h(U)/\delta^2)\,d\sigma(U)  \nonumber \\  &= \frac{1}{(2\pi)^NN!} \int_{\rr^N} \psi\biggl(\sum_{j=1}^N \eta_j^2\biggr) \prod_{1\le j<k\le N} (\eta_j-\eta_k)^2\, d\eta_1\,d\eta_2\,\cdots\, d\eta_N\,.
\end{align}
Now suppose that $\psi$ is the function that is $1$ in $[0,1]$ and $0$ outside. If the above formula holds for this $\psi$, then this is the same as the claim of the lemma. Note that $\psi$ has compact support but $\psi$ is not smooth. Take two smooth  functions $\psi_1$ and $\psi_2$ with compact support such that $\psi_1\le \psi\le \psi_2$ everywhere. Then note that 
\begin{align*}
&\lim_{\delta \ra 0} \delta^{-N^2}\int_{U(N)}\psi_1(h(U)/\delta^2)\,d\sigma(U) \le \liminf_{\delta \ra 0} \delta^{-N^2}\int_{U(N)}\psi(h(U)/\delta^2)\,d\sigma(U) \\
&\le \limsup_{\delta \ra 0} \delta^{-N^2}\int_{U(N)}\psi(h(U)/\delta^2)\,d\sigma(U) \le \lim_{\delta \ra 0} \delta^{-N^2}\int_{U(N)}\psi_2(h(U)/\delta^2)\,d\sigma(U) \,.
\end{align*}
Applying \eqref{matrix1} with $\psi_1$ and $\psi_2$, and then letting $\psi_1$ and $\psi_2$ approach $\psi$ pointwise almost everywhere, we get the desired identity for $\psi$. 
\end{proof}

We are now ready to carry out the second step in the proof of Theorem \ref{smallballthm} and complete the proof. This step involves the joint density of the eigenvalues of a random matrix from the Gaussian Unitary Ensemble (GUE). Recall that a GUE matrix is a random Hermitian matrix whose entries on and above the diagonal are independent, with the following distributions. The entries on the diagonal are standard real Gaussian random variables. The off-diagonal entries are complex Gaussian, with real and imaginary parts independent, mean zero, and variance $1/2$. In other words, the entries of a GUE matrix $X= (x_{jk})_{1\le j,k\le N}$ may be represented as 
\[
x_{jk} =
\begin{cases}
y_{jj} &\text{ if } j=k,\\
(z_{jk}+\I w_{jk})/\sqrt{2} &\text{ if } j<k,\\
(z_{kj}-\I w_{kj})/\sqrt{2} &\text{ if } j>k,
\end{cases}
\]
where $(y_{jj})_{1\le j\le N}$, $(z_{jk})_{1\le j<k\le N}$ and $(w_{jk})_{1\le j<k\le N}$ are collections of i.i.d.~standard real Gaussian random variables. Since $X$ is a Hermitian matrix, its eigenvalues are real. It is not difficult to show that with probability one, there is no eigenvalue with multiplicity greater than one. Let $\lambda_1<\lambda_2<\cdots< \lambda_N$ be the eigenvalues of $X$. \cite{md63} showed that the joint probability density of $(\lambda_1,\ldots,\lambda_N)$ is 
\begin{equation}\label{guedensity}
\rho(t_1,\ldots, t_N) = \frac{e^{-\|t\|^2/2} \prod_{1\le j<k\le N} (t_j-t_k)^2}{(2\pi)^{N/2} \prod_{j=1}^{N-1} j!} 
\end{equation}
on the domain $\{(t_1,\ldots, t_N): t_1<\cdots< t_n\}$, where $\|t\|^2 := t_1^2+\cdots +t_N^2$. Again, the normalizing constant has a relatively easy derivation via the Selberg integral formula. For more on this topic, see \cite{mehta91} and \cite{fw08}.
\begin{proof}[Proof of Theorem \ref{smallballthm}]
As in Lemma \ref{matrixlmm1}, let $B(r)$ denote the ball of radius $r$ around the origin in $\rr^N$. Then note that
\begin{align}
&\int_{B(1)}  \prod_{1\le j<k\le N} (\eta_j-\eta_k)^2\, d\eta_1\,d\eta_2\,\cdots \,d\eta_N\nonumber \\
&= \lim_{\gamma\ra 0} \int_{B(1)} e^{-\gamma^2\|\eta\|^2/2} \prod_{1\le j<k\le N} (\eta_j-\eta_k)^2\, d\eta_1\,d\eta_2\cdots\, d\eta_N\,.\label{matrixeq1}
\end{align}
By the change of variable $\xi_i = \gamma\eta_i$, we get
\begin{align}
&\int_{B(1)} e^{-\gamma^2\|\eta\|^2/2} \prod_{1\le j<k\le N} (\eta_j-\eta_k)^2\, d\eta_1\,d\eta_2\cdots\, d\eta_N\nonumber \\
&= \gamma^{-N^2}\int_{B(\gamma)} e^{-\|\xi\|^2/2} \prod_{1\le j<k\le N} (\xi_j-\xi_k)^2\, d\xi_1\,d\xi_2\,\cdots\, d\xi_N\,.\label{matrixeq2}
\end{align}
Let $X$ be the random GUE matrix defined above. Then from the formula \eqref{guedensity} for the joint density of the eigenvalues of $X$, we get
\begin{align}\label{matrixeq3}
\pp\biggl(\sum_{j=1}^N \lambda_j^2 \le \gamma^2\biggr) &= \frac{1}{(2\pi)^{N/2} \prod_{j=1}^N j!} \int_{B(\gamma)} e^{-\|\xi\|^2/2} \prod_{1\le j<k\le N} (\xi_j-\xi_k)^2\, d\xi_1\,d\xi_2\,\cdots \,d\xi_N\,.
\end{align}
On the other hand, note that
\begin{align*}
\sum_{j=1}^N \lambda_j^2 &= \tr(X^2) = \tr(X^*X)=\sum_{j,k=1}^N |x_{jk}|^2\\
&= \sum_{j=1}^N y_{jj}^2 + \sum_{1\le j<k\le N} (z_{jk}^2 + w_{jk}^2)\,.
\end{align*}
The last expression, being a sum of $N^2$ i.i.d.~standard real Gaussian random variables, has a $\chi^2$ distribution with $N^2$ degrees of freedom. Therefore, using the well-known formula for the probability density function of such a random variable (see, for example, \cite{jkb95}), we get
\begin{align}\label{matrixeq4}
\pp\biggl(\sum_{j=1}^N \lambda_j^2 \le \gamma^2\biggr) &= \frac{1}{2^{N^2/2} \Gamma(N^2/2)}\int_0^{\gamma^2}x^{N^2/2-1} e^{-x/2}\, dx\,.
\end{align}
Making the change of variable $y = x/\gamma^2$, we get
\begin{align}\label{matrixeq5}
\int_0^{\gamma^2}x^{N^2/2-1} e^{-x/2}\, dx &= \gamma^{N^2} \int_0^1y^{N^2/2-1} e^{-\gamma^2 y/2} \, dy\,.
\end{align}
Combining Lemma \ref{matrixlmm1} with the equations \eqref{matrixeq1}, \eqref{matrixeq2}, \eqref{matrixeq3}, \eqref{matrixeq4} and \eqref{matrixeq5} gives
\begin{align*}
\lim_{\delta \ra 0}\frac{\sigma(\{U: \|I-U\|\le \delta\})}{\delta^{N^2}} &=\lim_{\gamma \ra0} \frac{\gamma^{-N^2}}{(2\pi)^N N!}\int_{B(\gamma)} e^{-\|\xi\|^2/2} \prod_{1\le j<k\le N} (\xi_j-\xi_k)^2\, d\xi_1\,d\xi_2\,\cdots \,d\xi_N\\
&= \lim_{\gamma \ra0} \frac{\gamma^{-N^2}(2\pi)^{N/2}\prod_{j=1}^N j!}{(2\pi)^N N!}\pp\biggl(\sum_{j=1}^N \lambda_j^2 \le \gamma^2\biggr)\\
&= \lim_{\gamma \ra0} \frac{\prod_{j=1}^{N-1} j!}{(2\pi)^{N/2} 2^{N^2/2}\Gamma(N^2/2)}\int_0^1y^{N^2/2-1} e^{-\gamma^2 y/2} \, dy\\
&= \frac{\prod_{j=1}^{N-1} j!}{(2\pi)^{N/2} 2^{N^2/2}\Gamma(N^2/2)(N^2/2)}\,.
\end{align*}
This completes the proof of Theorem \ref{smallballthm}.
\end{proof}
Theorem \ref{smallballthm} has the following corollary, which will be useful in the next section.  Recall that the Hilbert--Schmidt norm of a unitary matrix of  order $N$ is always $\sqrt{N}$, and therefore the diameter of $U(N)$ under the metric induced by the Hilbert--Schmidt norm is $2\sqrt{N}$.
\begin{cor}\label{smallball}
Let $\sigma$ denote the Haar measure on $U(N)$. There are positive constants $C_1$ and $C_2$ depending only on $N$, such that for any  $\delta \in (0,\sqrt{N})$, 
\[
C_1 \delta^{N^2} \le \sigma(\{U: \|I-U\|\le \delta\}) \le C_2 \delta^{N^2}\,.
\]
\end{cor}
\begin{proof}
By Theorem \ref{smallballthm}, the inequalities are true for small enough $\delta$, where the threshold depends only on $N$. If $\delta$ is above this threshold, then the upper bound is true anyway since the Haar measure of any set is bounded by one. To generalize the lower bound to all $\delta\in (0,\sqrt{N})$, simply observe that the Haar measure of the ball of radius $\delta$ increases with $\delta$.  
\end{proof}


\section{A lower bound for the partition function}\label{lowbdsec}
Using the small ball probability lower bound derived in Section \ref{smallsec}, we will now obtain a lower bound for the partition function of $U(N)$ lattice gauge theory in a box. 
\begin{thm}\label{lowlmm}
Suppose that $\beta \ge 2$. Then there is a positive constant $C$ depending only on $N$ and $d$, such that for any $n$,
\[
Z(B_n, g_0)\ge e^{-Cn^d\log \beta}\,.
\]
\end{thm}
We need some preparatory lemmas before proving Theorem \ref{lowlmm}.
\begin{lmm}\label{phihslmm}
The function $\phi:U(N)\ra \rr$ defined in \eqref{phidef} satisfies the identity
\[
\phi(U) = \frac{1}{2}\|I-U\|^2\,.
\]
\end{lmm}
\begin{proof}
Take any $U\in U(N)$ and let $e^{\I \theta_1},\ldots, e^{\I \theta_N}$ be the eigenvalues of $U$, repeated by multiplicities and arranged in some arbitrary order. Then note that
\begin{align*}
\phi(U) &= \Re(\tr(I-U)) = \sum_{j=1}^N (1-\Re(e^{\I\theta_j}))\\
&= \sum_{j=1}^N(1-\cos\theta_j)\,.
\end{align*}
On the other hand, by \eqref{traceform},  
\begin{align*}
\|I-U\|^2 &= \sum_{j=1}^N 2(1-\cos\theta_j)\,.
\end{align*}
This completes the proof of the lemma.
\end{proof}
\begin{lmm}\label{uhslmm}
For any $U\in U(N)$ and any $N\times N$ complex matrix $M$,
\[
\|UM\| = \|M\|\,.
\]
\end{lmm}
\begin{proof}
Let $m_1,\ldots, m_N$ be the columns of $M$. Then 
\[
\|UM\|^2 = \|Um_1\|^2 + \cdots + \|Um_N\|^2\,.
\]
Since $U$  is a unitary matrix, $\|Ux\|^2 = \|x\|^2$ for any $x$. Since 
\[
\|M\|^2 = \|m_1\|^2+ \cdots + \|m_N\|^2\,,
\]
this completes the proof.
\end{proof}
\begin{lmm}\label{uprodlmm}
For any $U_1,\ldots, U_m\in U(N)$,
\[
\|I - U_1U_2\cdots U_m\| \le \sum_{j=1}^m \|I- U_j\|\,.
\]
\end{lmm}
\begin{proof}
By the triangle inequality for the Hilbert--Schmidt norm,
\begin{align*}
\|I - U_1U_2\cdots U_m\| &\le \sum_{j=0}^{m-1}\|U_1\cdots U_j - U_1\cdots U_{j+1}\|\\
&= \sum_{j=0}^{m-1}\|U_1\cdots U_j (I-U_{j+1})\|\,.
\end{align*}
To complete the proof, note that by  Lemma \ref{uhslmm}, $\|U_1\cdots U_j (I-U_{j+1})\| = \|I-U_{j+1}\|$.
\end{proof}
\begin{lmm}\label{uinvlmm}
For any $U\in U(N)$, $\|I-U^{-1}\| = \|I-U\|$.
\end{lmm}
\begin{proof}
Note that $I-U^{-1}$ = $-U^{-1}(I-U)$, and apply Lemma \ref{uhslmm}.
\end{proof}
We are now ready to prove Theorem \ref{lowlmm}.
\begin{proof}[Proof of Theorem \ref{lowlmm}]
Take any $\delta \in (0,\sqrt{N})$. Suppose that $U\in U(B_n)$ is a configuration such that for every $(x,y)\in E_n$,
\[
\|I-U(x,y)\|\le \delta\,.
\]
Then by Lemma \ref{uinvlmm}, $\|I-U(y,x)\|$ is also bounded by $\delta$. Consequently, for any plaquette $(x,j,k)\in B_n'$, Lemma \ref{uprodlmm} implies that
\begin{align*}
\|I-U(x,j,k)\| \le 4\delta\,.
\end{align*}
Therefore by Lemma \ref{phihslmm},
\[
\phi(U(x,j,k)) \le 8\delta^2\,.
\]
Summing over all plaquettes, this gives
\[
S_{B_n}(U) = \sum_{(x,j,k)\in B_n'} \phi(U(x,j,k)) \le 8|B_n'|\delta^2\,.
\]
Thus, by the lower bound from Corollary \ref{smallball},
\begin{align*}
\sigma_{B_n}(\{U: S_{B_n}(U)\le 8|B_n'|\delta^2\}) &\ge \sigma_{B_n}(\{U: \|I-U(x,y)\| \le \delta \text{ for all } (x,y)\in E_n\})\\
&\ge (C_1\delta)^{C_3n^d}\,.
\end{align*}
Choosing $\delta$ such that $8\delta^2 = 1/\beta$, we get
\begin{align*}
Z(B_n,g_0)&= \int_{U(B_n)} e^{-\beta S_{B_n}(U)} \,d\sigma_{B_n}(U)\\
 &\ge e^{-|B_n'|} \sigma_{B_n}(\{U: S_{B_n}(U)\le |B_n'|/\beta\})\\
 &\ge e^{-Cn^d\log \beta}\,,
\end{align*}
where we used the assumption that $\beta \ge 2$ in the last step, choosing $C$ sufficiently large. 
This completes the proof of the theorem.
\end{proof}


\section{Smallness of the Wilson action}\label{smallnesssec}
The following theorem shows that if $\beta$ is large, then under the lattice gauge measure, the Wilson action is at most of order $n^d\beta^{-1}\log \beta$. Interestingly, this result is almost a direct consequence of the lower bound on the partition function obtained in Section \ref{lowbdsec}.
\begin{thm}\label{hbnupper}
Suppose that $\beta \ge2$. Then there is a positive constant $C$    depending only on $N$ and $d$ such that for any $t> 0$ and $n\ge 2$,
\begin{align*}
\mu_{B_n,g_0}(\{U: \beta S_{B_n}(U) \ge Cn^d\log \beta + t\})&\le e^{-t}\,.
\end{align*}
\end{thm}
\begin{proof}
Note that by Theorem \ref{lowlmm}, there is a constant $C$ depending only on $N$ and $d$ such that
\begin{align*}
\int_{U(B_n)} e^{\beta S_{B_n}(U)}\, d\mu_{B_n, g_0}(U) &= \frac{1}{Z(B_n, g_0)}\le e^{Cn^d\log \beta}\,.
\end{align*}
Thus, for any $t>0$,
\begin{align*}
&\mu_{B_n,g_0}(\{U: \beta S_{B_n}(U) \ge Cn^d\log \beta + t\})\\
&\le e^{-Cn^d\log \beta-t}\int_{U(B_n)} e^{\beta S_{B_n}(U)}\, d\mu_{B_n, g_0}(U)\le e^{-t}\,,
\end{align*}
which proves the claim.
\end{proof}
This theorem has the following important corollary, which shows that if $\beta$ is large, then the computation of the partition function can be reduced to an integral over configurations that nearly minimize the Wilson action.
\begin{cor}\label{smallcor}
Suppose that $\beta \ge 2$. Then there is a constant $C_0$ depending only on $N$ and $d$ such that for any $n\ge 2$,
\[
Z(B_n, g_0) \le 2\int_{A} e^{-\beta S_{B_n} (U)}\, d\sigma_{B_n}(U) \le 2Z(B_n, g_0)\,,
\]
where
\[
A = \{U: \beta S_{B_n}(U) \le C_0n^d\log \beta\}\,.
\]
\end{cor}
\begin{proof}
The upper bound is obvious. For the lower bound, observe that
\[
\frac{1}{Z(B_n,g_0)}\int_{A} e^{-\beta S_{B_n} (U)} \,d\sigma_{B_n}(U) = \mu_{B_n,g_0}(A)\,,
\]
and apply Theorem \ref{hbnupper} to show that $\mu_{B_n,g_0}(A)\ge 1/2$ if $C_0$ is sufficiently large.
\end{proof}


\section{Axial gauge fixing}\label{gaugesec}
Let $G(B_n)$ be the set of all maps from $B_n$ into $U(N)$. In other words, each element $G\in G(B_n)$ assigns a unitary matrix $G(x)$ to each vertex $x\in B_n$. The set $G(B_n)$ is a group under componentwise multiplication, and this group acts on the set $U(B_n)$ in the following way. For each $U\in U(B_n)$ and $G\in G(B_n)$, let $V= GU$ be defined as
\[
V(x,y) = G(x) U(x,y) G(y)^{-1}\,.
\]
It is easy to see that this is indeed a group action. The action of any $G\in G(B_n)$ is called a gauge transform, and $G(B_n)$ is the set of all possible gauge transforms. A key property of the Wilson action is that it is invariant under gauge transforms, as shown by the following proposition.
\begin{prop}\label{gaugeprop1}
Let $f:U(B_n) \ra \rr$ be a function such that $f(U)$ depends on $U$ only through the plaquette variables $(\phi(U(x,j,k)))_{(x,j,k)\in B_n'}$. Then for any $U\in U(B_n)$ and $G\in G(B_n)$, $f(GU) = f(U)$. In particular, $S_{B_n}(GU)=S_{B_n}(U)$.
\end{prop}
\begin{proof}
Let $V= GU$. A simple verification shows that for any $(x,j,k)\in B_n'$,
\begin{align*}
V(x,j,k) = G(x)U(x,j,k)G(x)^{-1}\,.
\end{align*}
Therefore $\tr(V(x,j,k)) = \tr(U(x,j,k))$, which implies that $\phi(V(x,j,k)) = \phi(U(x,j,k))$ and hence $f(V) = f(U)$. 
\end{proof}
Two configurations $U$ and $V$ are said to be gauge equivalent if $V= GU$ for some gauge transform $G$. Roughly speaking, gauge fixing refers to any procedure that, given a configuration $U$, produces a gauge equivalent configuration $V$ with some desirable properties. A popular gauge fixing process, known as axial gauge fixing, goes as follows. Recall the definitions of $E_n$, $E_n^0$ and $E_n^1$ from Section~\ref{results}. Define
\[ 
U_0(B_n) := \{U\in U(B_n): U(x,y)=I \text{ for all } (x,y)\in E_n^0\}\,.
\] 
Given $U\in U(B_n)$, define an element $G_U\in G(B_n)$ as follows. For any $x,y\in B_n$, write $x\prec y$ if $x$ comes before $y$ in the lexicographic ordering. The function $G_U$ is defined inductively. Let $G_U(0)=I$. Take any $x = (x_1,\ldots, x_d)\in B_n$ and suppose that $G_U(y)$ has been defined for all $y\prec x$. Let $j$ be the largest index such that $x_j\ne 0$. Let $y:= x- e_j$. 
Then $y\in B_n$ and $y\prec x$. Thus, $G_U(y)$ is already defined. Let 
\[
G_U(x) := G_U(y) U(y,x)\,.
\]
In this way, $G_U(x)$ gets defined for all $x\in B_n$. 
\begin{prop}\label{gaugeprop2}
Take any $U\in U(B_n)$ and define $G_U$ in the above manner. Then $G_UU\in U_0(B_n)$. 
\end{prop}
\begin{proof}
Let $V := G_U U$. Take any $(x,y)\in E_n^0$. Then by the definition of $G_U$, it is clear that 
\[
G_U(y) = G_U(x)U(x,y)\,.
\]
Thus,
\begin{align*}
V(x,y) &= G_U(x) U(x,y) G_U(y)^{-1} = I\,,
\end{align*}
which proves that $V\in U_0(B_n)$.
\end{proof}
A property of axial gauge fixing that will be important in this paper is that if $U$ is a random configuration drawn from the product measure $\sigma_{B_n}$, then $G_UU$ retains the product structure. This is made precise in the following lemma, which is stated and proved in a probabilistic language.
\begin{lmm}\label{gaugeind}
Let $U$ be a random configuration drawn from the product measure $\sigma_{B_n}$. Let $V := G_U U$, where $G_U$ is defined as above. Then the matrices $\{V(x,y): (x,y)\in E_n^1\}$ are independent and Haar-distributed.
\end{lmm}
\begin{proof}
Let $U_0 := \{U(x,y): (x,y)\in E_n^0\}$ and $U_1 := \{U(x,y): (x,y)\in E_n^1\}$. Note that $U_0$ and $U_1$ are disjoint collections of independent Haar-distributed matrices. It is easy to see by induction that each $G_U(x)$ is a product of some elements of $U_0$. Therefore, if we condition on $U_0$, then $\{V(x,y):(x,y)\in E_n^1\}$ is a collection of independent random matrices. Moreover, conditional on $U_0$, each $V(x,y)$ is Haar-distributed. Thus, given $U_0$, the matrices $\{V(x,y): (x,y)\in E_n^1\}$ are independent and Haar-distributed. Since the conditional distribution of this collection does not depend on $U_0$, it is also the unconditional distribution. This proves the claim of the lemma.
\end{proof}
Let $\sigma^0_{B_n}$ be the probability law of the configuration $V$ in Lemma \ref{gaugeind}. Lemma \ref{gaugeind} has the following corollary, which allows us to reduce the integration of a function of the Wilson action over $U(B_n)$  to an integration over $U_0(B_n)$. 
\begin{cor}\label{uu0}
Let $U_0(B_n)$ and $\sigma^0_{B_n}$ be defined as above. Then for any bounded Borel  measurable $f:U(B_n) \ra \rr$  such that $f(U)$ depends on $U$ only through $(\phi(U(x,j,k)))_{(x,j,k)\in B_n'}$, 
\[
\int_{U(B_n)} f(U)\, d\sigma_{B_n}(U) = \int_{U_0(B_n)} f(U) \,d\sigma^0_{B_n}(U)\,.
\]
\end{cor}
\begin{proof}
By Proposition \ref{gaugeprop1} and Proposition \ref{gaugeprop2}, for any $U\in U(B_n)$, $G_UU\in U_0(B_n)$ and
\[
f(U) = f(G_UU)\,.
\]
Thus,
\begin{align*}
\int_{U(B_n)} f(U) \,d\sigma_{B_n}(U) &= \int_{U(B_n)} f(G_UU)\, d\sigma_{B_n}(U)\,.
\end{align*}
By Lemma \ref{gaugeind}, if $U$ is a random configuration with law $\sigma_{B_n}$, then $G_U U$ is a random configuration with law $\sigma^0_{B_n}$. Therefore, 
\begin{align*}
\int_{U(B_n)} f(G_U U) \,d\sigma_{B_n}(U) &= \int_{U_0(B_n)} f(V) \,d\sigma^0_{B_n}(V)\,.
\end{align*}
Combined with the previous identity, this proves the claim.
\end{proof}


\section{An upper bound for the partition function}\label{uppersec}
The following result shows that while computing an upper bound for the free energy per site in $B_n$, if $\beta$ is large and $n$ is not too large (depending on $\beta$), it suffices to restrict attention to configurations where all matrices are close to the identity. The restriction that $n$ needs to be sufficiently small will be removed later. 
\begin{thm}\label{mainstep1}
There is a constant $C_1$ depending only on $N$ and $d$ such that the following is true. For any $n\ge 2$, 
\begin{align*}
F(B_n, g_0) &\le \frac{\log 2}{n^d} + \frac{1}{n^{d}}\log \int_{U_0^\beta(B_n)}  e^{-\beta S_{B_n}(U)}\, d\sigma^0_{B_n}(U)\,,
\end{align*}
where
\[
U_0^\beta(B_n) := \biggl\{U\in U_0(B_n): \|I-U(x,y)\|\le C_1n^{(d+1)/2}\biggl(\frac{\log\beta}{\beta}\biggr)^{1/2}\textup{ for all } (x,y)\in E_n\biggr\}\,.
\]
\end{thm}
For $x\in \zz^d$, let $|x|_1$ denote the $\ell^1$ norm of $x$, that is, the sum of the absolute values of the coordinates of $x$. We need the following lemma, which shows that if the Wilson action of a configuration $U\in U_0(B_n)$ is small, then $U(x,y)$ is close to $I$ for every $(x,y)\in E_n$ such that $|x|_1$ is not too large. One may call this a discrete nonlinear Poincar\'e inequality. 
\begin{lmm}\label{p1lmm}
Take any $n\ge 2$. For any $U\in U_0(B_n)$ and any $(x,y)\in E_n$,
\[
\|I-U(x,y)\|\le (2|x|_1 S_{B_n}(u))^{1/2}\,.
\] 
\end{lmm}
\begin{proof}
We will prove by induction on $|x|_1$ that for every $(x,y)\in E_n$, $\|I-U(x,y)\|$ is bounded above by the sum of $\|I-U(z,j,k)\|$ over $(z,j,k)\in B(x)$, where $B(x)$ is a subset of $B_n'$ of size $\le |x|_1$. This is clearly true if $|x|_1=0$, since every edge incident to the origin belongs to $E_n^0$. Now take any $(x,y)\in E_n$ and suppose that the claim has been proved for every $(x',y')\in E_n$ with $|x'|_1< |x|_1$. Let $x_1,\ldots,x_d$ be the coordinates of $x$. Then $y = x+ e_j$ for some $j$. Let $k$ be the largest index such that $x_k\ne 0$. If $k\le j$ then $(x,y)\in E_n^0$, which is the trivial case. Therefore assume that $k >j$. Let $z := x-  e_k$. Then $(z,j,k)\in B_n'$. Note that the edges $(z, z+ e_k)$ and $(z+e_j, z + e_j + e_k)$ belong to $E_n^0$. Therefore
\begin{align*}
U(z,j,k) &= U(z, z+ e_j) U(z+ e_j + e_k, z+ e_k)\\
&= U(z, z+ e_j) U(x,y)^{-1}\,.
\end{align*}
By Lemma \ref{uhslmm}, this gives
\begin{align*}
\|I-U(x,y)\| &= \|I- U(z,j,k)^{-1}U(z,z+ e_j)\|\\
&= \|U(z,j,k)^{-1}(U(z,j,k)-U(z,z+ e_j))\|\\
&= \|U(z,j,k) - U(z,z+ e_j)\|\\
&\le \|I-U(z,j,k)\| + \|I- U(z,z+ e_j)\|\,.
\end{align*}
Since $|z|_1= |x|_1-1$, this completes the induction step. Thus, for any $(x,y)\in E_n$, there exists a set $B(x)\subseteq B_n'$ of size $\le |x|_1$ such that
\[
\|I-U(x,y)\|\le \sum_{(z,j,k)\in B(x)} \|I-U(z,j,k)\|\,.
\]
Applying the Cauchy--Schwarz inequality and Lemma \ref{phihslmm}, this gives
\begin{align*}
\|I-U(x,y)\|&\le \biggl(|B(x)|\sum_{(z,j,k)\in B(x)} \|I-U(z,j,k)\|^2\biggr)^{1/2} \\
&\le (2|x|_1 S_{B_n}(U))^{1/2}\,.
\end{align*}
This completes the proof of the lemma. 
\end{proof}
We are now ready to prove Theorem \ref{mainstep1}.
\begin{proof}[Proof of Theorem \ref{mainstep1}]
Take any $n\ge 2$. Let 
\[
A := \{U\in U(B_n): \beta S_{B_n}(U) \le C_0n^d\log \beta\}\,,
\]
where $C_0$ is as in Corollary \ref{smallcor}, so that
\begin{equation}\label{zineq1}
Z(B_n, g_0)\le 2\int_A e^{-\beta S_{B_n}(U)} \,d\sigma_{B_n}(U)\,.
\end{equation}
By Corollary \ref{uu0},
\begin{align}\label{zineq2}
\int_A e^{-\beta S_{B_n}(U)} d\sigma_{B_n}(U) = \int_{A_0} e^{-\beta S_{B_n}(U)} \,d\sigma^0_{B_n}(U)\,,
\end{align}
where
\[
A_0:= A\cap U_0(B_n)\,.
\]
For any $U\in A_0$ and $(x,y)\in E_n$, Lemma \ref{p1lmm} gives
\begin{align*}
\|I-U(x,y)\|&\le  (2|x|_1 S_{B_n}(U))^{1/2}\\
&\le Cn^{(d+1)/2} \biggl(\frac{\log \beta}{\beta}\biggr)^{1/2}\,.
\end{align*}
Therefore, if the constant $C_1$ in the statement of the theorem is chosen appropriately, then
\begin{align*}
A_0\subseteq U_0^\beta(B_n)\,.
\end{align*}
By \eqref{zineq1} and \eqref{zineq2}, this completes the proof of the theorem.
\end{proof}


\section{From Lie group to Lie algebra}\label{liesec}
The purpose of this section is to lay the groundwork for replacing the integrals over the Lie group $U(N)$ with integrals over its Lie algebra $\mathfrak{u}(N)$, the set of all skew-Hermitian matrices of order $N$. In reality, we will be integrating over $\I$ times $\mathfrak{u}(N)$, that is, the set of  Hermitian matrices of order~$N$. The results of Section \ref{smallsec} are of crucial importance in this maneuver. 

Let $H(N)$ be the vector space of all $N\times N$ complex Hermitian matrices, equipped with the Hilbert--Schmidt norm. Take any $H\in H(N)$. Write the $(j,k)^{\mathrm{th}}$ entry of $H$ as  
\[
x_{jk} =
\begin{cases}
y_{jj} &\text{ if } j=k,\\
(z_{jk}+\I w_{jk})/\sqrt{2} &\text{ if } j<k,\\
(z_{kj}-\I w_{kj})/\sqrt{2} &\text{ if } j>k,
\end{cases}
\]
The $N^2$ parameters $(y_{jj})_{1\le j\le N}$, $(z_{jk})_{1\le j<k\le N}$ and $(w_{jk})_{1\le j<k\le N}$ define an isometry between $H(N)$ with Hilbert--Schmidt norm and $\rr^{N^2}$ with Euclidean norm. Note that the $\sqrt{2}$ in the above representation was inserted to guarantee the isometric nature of the correspondence. This isometry between the two spaces gives a natural definition of Lebesgue measure on $H(N)$, which we will denote by $\lambda$.

For any $U\in U(N)$ and $r> 0$, let 
\[
B(U, r) := \{V\in U(N): \|U-V\|\le r\}\,,
\]
and for any $H\in H(N)$ and $r > 0$, let
\[
b(H, r) := \{G\in H(N): \|H-G\|\le r\}\,.
\]
Recall that for any $H\in H(N)$, $e^{\I H} \in U(N)$, where
\[
e^{\I H} := \sum_{j=0}^\infty \frac{1}{j!} (\I H)^j\,.
\]
The goal of this section is to prove the following theorem. 
\begin{thm}\label{liethm}
Let $\sigma$ denote the Haar measure on $U(N)$. There exists $r_0>0$, depending only on $N$, such that the following is true. Let $f: U(N)\ra [0,\infty)$ be a Borel measurable function. Then for any $r\le r_0$, 
\begin{align*}
(2-e^{6r})^{N^2} C_N\int_{b(0,r)} f(e^{\I H}) \,d\lambda(H)\le \int_{B(I, 2r)} f(U) \,d\sigma(U) &\le e^{6rN^2}C_N\int_{b(0,3r)} f(e^{\I H})\, d\lambda(H)\,,
\end{align*}
where
\[
C_N := \frac{\prod_{j=1}^{N-1} j!}{(2\pi)^{N(N+1)/2}}\,.
\]
More generally, if for some $n\ge 1$,  $f:U(N)^n\ra[0,\infty)$ is a Borel measurable function, then for any $r\le r_0$, 
\begin{align*}
&(2-e^{6r})^{N^2n} C_N^n\int_{b(0,r)^n} f(e^{\I H_1},\ldots, e^{\I H_n}) \,d\lambda(H_1)\, \cdots \, d\lambda(H_n) \\
&\le \int_{B(I, 2r)^n} f(U_1,\ldots,U_n) \,d\sigma(U_1)\,\cdots \, d\sigma(U_n) \\
&\le e^{6rN^2n}C_N^n\int_{b(0,3r)^n} f(e^{\I H_1},\ldots, e^{\I H_n}) \,d\lambda(H_1)\, \cdots \, d\lambda(H_n)\,.
\end{align*}
\end{thm}
Incidentally, the reciprocal of the constant $C_N$ equals the volume of $U(N)$, when $U(N)$ is considered as a submanifold of $\rr^{N^2}$. This has been noted recently in \cite{df16}, where it is derived as consequence of a formula of \cite{hurwitz}. It would be interesting to see if the derivation in \cite{df16} can yield an alternative proof of Theorem \ref{liethm}.

Another remark, pointed out to me by Len Gross, is that Theorem \ref{liethm} can probably be generalized to arbitrary compact Lie subgroups of $U(N)$ using general properties of the exponential map of $U(N)$. The same remark applies to the results of Section \ref{smallsec}.

Several lemmas are needed for the proof of Theorem \ref{liethm}. We will also need Theorem \ref{smallballthm} from Section \ref{smallsec}.  For easy reference, let
\[
\psi(H) := e^{\I H}\,.
\]
The first lemma gives an important set of inequalities for $\psi$.
\begin{lmm}\label{psilmm}
Suppose that $H_1, H_2\in b(0,r)$. Then 
\[
(2-e^r) \|H_1-H_2\|\le \|\psi(H_1)-\psi(H_2)\|\le e^r \|H_1-H_2\|\,.
\]
In particular, $\psi$ is a continuous function.
\end{lmm}
\begin{proof}
For any $j\ge 1$,
\begin{align*}
\|H_1^j - H_2^j\|&= \|(H_1^j - H_1^{j-1} H_2) + (H_1^{j-1} H_2 - H_1^{j-2} H_2^2) + \cdots + (H_1H_2^{j-1} - H_2^j)\|\\
 &\le \sum_{k=1}^{j} \|H_1^{j-k}(H_1-H_2)H_2^{k-1}\|\\
 &\le \sum_{k=1}^{j} \|H_1\|^{j-k}\|H_1-H_2\|\|H_2\|^{k-1}\\
 &\le jr^{j-1} \|H_1-H_2\|\,.
\end{align*}
Thus,
\begin{align*}
\|e^{\I H_1} - e^{\I H_2}\| &\le \sum_{j=1}^\infty\frac{1}{j!}\|H_1^j - H_2^j\|\\
&\le \sum_{j=1}^\infty\frac{jr^{j-1}}{j!} \|H_1-H_2\|\\
&= e^r \|H_1-H_2\|\,,
\end{align*}
which proves the upper bound. Next, note that
\begin{align*}
\|e^{\I H_1} - e^{\I H_2}\| &= \biggl\| \sum_{j=1}^\infty \frac{\I^j}{j!}(H_1^j - H_2^j)\biggr\|\\
&\ge \|H_1-H_2\| - \biggl\| \sum_{j=2}^\infty \frac{\I^j}{j!}(H_1^j - H_2^j)\biggr\|\\
&\ge \|H_1-H_2\| - \sum_{j=2}^\infty \frac{1}{j!}\|H_1^j - H_2^j\|\\
&\ge \|H_1-H_2\| - \sum_{j=2}^\infty \frac{jr^{j-1}}{j!}\|H_1 - H_2\|\\
&= (2-e^r)\|H_1-H_2\|\,.
\end{align*}
This proves the lower bound and completes the proof of the lemma.
\end{proof}
The above lemma has two important corollaries.
\begin{cor}\label{psicor0}
There exists $r_0>0$, depending only on $N$, such that the following is true. For any $r\le r_0$,  $H\in b(0,r)$ and $\delta \le r$,
\begin{align*}
B(e^{\I H}, (2-e^{6r}) \delta) \subseteq \psi(b(H,\delta)) \subseteq B(e^{\I H}, e^{2r} \delta)\,.
\end{align*}
\end{cor}
\begin{proof}
Without loss of generality, $r$ is so small that $e^{6r} < 2$. Take any $G\in b(H, \delta)$. Then $G\in b(0,r+\delta)$. Therefore, by Lemma \ref{psilmm},
\begin{align*}
\|e^{\I H} - e^{\I G}\| &\le e^{r+\delta} \|H-G\|\le e^{2r} \delta\,.
\end{align*}
This proves one inclusion. Next, take any 
\[
U\in B(e^{\I H}, (2-e^{6r})\delta)\,.
\]
Let $e^{\I\theta_1},\ldots, e^{\I\theta_N}$ be the eigenvalues of $U$, with $\theta_j$'s chosen such that $-\pi\le \theta_j < \pi$ for each $j$. Then $U$ has a spectral decomposition $VDV^*$, where $V$ is a unitary matrix and $D$ is a diagonal matrix with diagonal entries  $e^{\I\theta_1},\ldots, e^{\I\theta_N}$. Let $\Lambda$ be the diagonal matrix with diagonal entries $\theta_1,\ldots, \theta_N$. Then $G := V\Lambda V^*$ is a Hermitian matrix, and $U = e^{\I G}$. 

Now note that by the upper bound from Lemma \ref{psilmm}, 
\[
\|I-e^{\I H}\|= \|e^{\I 0} - e^{\I H}\| \le e^r \|H\|\le e^r r\,.
\]
In other words, $e^{\I H}\in B(I, e^rr)$. Consequently,
\[
B(e^{\I H}, \delta) \subseteq B(I, e^r r + \delta)\subseteq B(I, (e^r+1)r)\subseteq B(I, 3r)\,.
\]
Since $U\in B(e^{\I H},\delta)$, the above inclusion and the identity \eqref{traceform} imply that
\begin{align*}
9r^2\ge \|I-U\|^2 = 2\sum_{j=1}^N (1-\cos\theta_j)\,.
\end{align*}
Now recall that the $\theta_j$'s are all in $[-\pi,\pi]$. Therefore, if $r$ is sufficiently small (depending on $N$), the above inequality implies that $1-\cos\theta_j \ge \theta_j^2/8$ for each $j$. As a consequence,
\[
\|G\|^2 = \sum_{j=1}^N\theta_j^2\le 36r^2\,.
\]
Thus, $G$ and $H$ both belong to $b(0,6r)$. By the lower bound from Lemma \ref{psilmm}, this gives
\[
\|G-H\|\le \frac{\|U-e^{\I H}\|}{(2-e^{6r})}\le \delta\,. 
\]
Therefore $U\in \psi(b(H,\delta))$. This completes the proof of the lemma.
\end{proof}
\begin{cor}\label{psicor}
For any $r< \log 2$, $\psi$ is injective on $b(0,r)$, and $\psi^{-1}$ is continuous on $\psi(b(0,r))$.
\end{cor}
\begin{proof}
The injectivity on $b(0,r)$ for $r<\log 2$ is clear from the lower bound of Lemma \ref{psilmm}. Continuity of $\psi^{-1}$ follows also from the same lower  bound.
\end{proof}
Fix some $r_0$ so small that the conclusions of Corollary \ref{psicor0} and Corollary \ref{psicor} are valid for $r\le 2r_0$. For any Borel set $A\subseteq b(0,r_0)$, let 
\[
\nu(A):= \sigma(\psi(A))\,.
\]
Note that $\nu$ is well-defined, since by Corollary \ref{psicor}, $\psi^{-1}$ is a measurable map on $\psi(b(0,r_0))$ and therefore $\psi(A)$ is a Borel subset of $U(N)$ for any Borel set $A\subseteq b(0,r_0)$. Next, note that since $\psi$ is injective on $b(0,r_0)$, $\nu$ is countably additive. Thus, $\nu$ is a measure on $b(0,r_0)$. Extend the measure $\nu$ to the whole of $H(N)$ by defining it to be zero outside $b(0,r_0)$. Clearly, $\nu$ is a finite measure, with total mass depending only on $N$. 
\begin{lmm}\label{hsinvprop}
For any $V\in U(N)$ and $\delta >0$,
\[
\sigma(B(V, \delta)) = \sigma(B(I,\delta))\,.
\]
\end{lmm}
\begin{proof}
Lemma \ref{uhslmm} implies that for any $U,V\in U(N)$,
\[
\|V-U\|_{\mathrm{HS}} = \|V(I-V^*U)\|_{\mathrm{HS}} = \|I-V^*U\|_{\mathrm{HS}}\,.
\]
By the invariance of the Haar measure under multiplication, the set of all $U$ such that $\|I-U\|_{\mathrm{HS}}\le \delta$ has the same measure as the set of all $U$ such that $\|I-V^*U\|_{\mathrm{HS}}\le \delta$.
\end{proof}
\begin{lmm}\label{abscont1}
For any $H\in H(N)$ and $\delta >0$, 
\[
\nu(b(H, \delta))\le C\delta^{N^2}\,,
\]
where $C$ depends only on $N$.
\end{lmm}
\begin{proof}
Since the total mass of $\nu$ is bounded by a finite constant that depends only on $N$, it suffices to prove the lemma for sufficiently small $\delta$. Assume that $\delta \le r_0$. If $\|H\| > 2r_0$, $b(H, \delta)$ does not intersect $b(0,r_0)$, and therefore $\nu(b(H,\delta))=0$. So assume that $\|H\|\le 2r_0$. Then by Corollary~\ref{psicor0}, 
\[
\psi(b(H,\delta)) \subseteq B(e^{\I H}, e^{4r_0} \delta)\,.
\]
Note that for any Borel set $A\subseteq H(N)$ and any Borel set $B\supseteq \psi(A)$,
\[
\nu(A) = \nu(A \cap b(0,r_0)) = \sigma (\psi(A\cap b(0,r_0)))\le \sigma(B)\,.
\]
Combining these two observations, we get
\[
\nu(b(H,\delta)) \le \sigma(B(e^{\I H}, e^{4r_0} \delta))\,.
\]
By Corollary \ref{smallball} and Lemma \ref{hsinvprop}, we get the desired upper bound on $\sigma(B(e^{\I H}, e^{4r_0} \delta))$.
\end{proof}
\begin{lmm}\label{abscont2}
The measure $\nu$ is absolutely continuous with respect to the Lebesgue measure $\lambda$ on $H(N)$ that was defined at the beginning of this section.
\end{lmm}
\begin{proof}
Take any Borel set $A\subseteq H(N)$ of Lebesgue measure zero. Take any $\eta > 0$. From the standard construction of Lebesgue measure (for example, in Chapter 11 of \cite{rudin}), it follows that there exists a countable collection of Euclidean balls $B_1, B_2,\ldots$ such that 
\[
A \subseteq \bigcup_{j=1}^\infty B_j
\]
and
\[
\sum_{j=1}^\infty \lambda(B_j) <\eta\,.
\]
On the other hand, by Lemma \ref{abscont1},  $\nu(B_j) \le C \lambda(B_j)$ 
for each $j$. Thus, $\nu(A)\le C\eta$. Since this is true for any $\eta$, $\nu(A)$ must be zero.
\end{proof}
We are now ready to prove Theorem \ref{liethm}.
\begin{proof}[Proof of Theorem \ref{liethm}]
Lemma \ref{abscont2} implies the existence of a Radon--Nikodym derivative $\rho$ of $\nu$ with respect to $\lambda$. By standard results (for example, Theorem 7.8 in \cite{rudin87}) we know that for almost every $H$, 
\begin{align*}
\rho(H) &= \lim_{\delta\ra 0} \frac{\nu(b(H, \delta))}{\lambda(b(H,\delta))}\,.
\end{align*}
Plugging in the formula for the volume of a Euclidean ball, this gives
\begin{align}\label{rhohform}
\rho(H) &= \lim_{\delta\ra 0} \frac{\Gamma(N^2/2+1)\nu(b(H, \delta))}{\pi^{N^2/2} \delta^{N^2}}
\end{align}
for almost all $H$. Now take any $H$ such that $\|H\|\le r< r_0$, for which the above identity holds. Then for $\delta$ sufficiently small, $b(H,\delta)\subseteq b(0,r_0)$, and therefore
\[
\nu(b(H,\delta)) = \sigma(\psi(b(H,\delta)))\,.
\] 
By Corollary \ref{psicor0}, this implies that
\begin{align*}
\sigma(B(e^{\I H}, (2-e^{6r}) \delta)) \le \nu(b(H,\delta)) \le \sigma(B(e^{\I H}, e^{2r} \delta))\,.
\end{align*}
Applying Lemma \ref{hsinvprop} and Theorem \ref{smallballthm}, we get
\begin{align*}
(2-e^{6r})^{N^2} C'_N \le \lim_{\delta\ra 0} \frac{\nu(b(H,\delta))}{\delta^{N^2}} \le e^{2rN^2}C_N'\,,
\end{align*}
where
\[
C_N' := \frac{\prod_{j=1}^{N-1} j!}{(2\pi)^{N/2} 2^{N^2/2}\Gamma(N^2/2+1)}\,.
\]
By \eqref{rhohform}, this implies that for almost all $H\in b(0,r)$,
\begin{equation}\label{rhoineqs}
(2-e^{6r})^{N^2} C_N \le \rho(H) \le e^{2rN^2}C_N\,,
\end{equation}
where $C_N$ is the constant defined in the statement of the theorem.

If $r\le r_0$, then by the definition of $\nu$ and the injectivity of $\psi$,
\begin{align*}
\int_{\psi(b(0,r))} f(U)\, d\sigma(U) &= \int_{b(0,r)}f(\psi(H))\, d\nu(H) = \int_{b(0,r)}f(\psi(H)) \rho(H) \,d\lambda(H)\,,
\end{align*}
and therefore by \eqref{rhoineqs},
\begin{align}
(2-e^{6r})^{N^2} C_N\int_{b(0,r)}f(\psi(H))\,  d\lambda(H)&\le \int_{\psi(b(0,r))} f(U) \,d\sigma(U)\nonumber \\
&\le e^{2rN^2}C_N\int_{b(0,r)}f(\psi(H))\,  d\lambda(H)\,.\label{last1}
\end{align}
Now note that if $r$ is sufficiently small, then by Corollary \ref{psicor0}, 
\begin{equation}\label{last2}
\psi(b(0,r))\subseteq B(I, 2r)\subseteq \psi(b(0,3r))\,. 
\end{equation}
To complete the proof for the first assertion of the theorem, use \eqref{last2} in \eqref{last1}, after replacing $r$ with $3r$ in the second inequality.

For the second assertion (that is, for $f:U(N)^n \ra [0,\infty)$), fix $U_2,\ldots, U_n$ and integrate over $U_1$. In this integral, apply the upper bound for the case $n=1$ to get
\begin{align*}
&\int_{B(I, 2r)^n} f(U_1,\ldots,U_n) \,d\sigma(U_1)\,\cdots \, d\sigma(U_n) \\
&\le e^{6rN^2}C_N\int_{b(0,3r)\times B(I,2r)^{n-1}} f(e^{\I H_1},U_2,\ldots, U_n) \,d\lambda(H_1)\,d\sigma(U_2)\, \cdots \, d\sigma(U_n)\,.
\end{align*}
Next, fix $H_1,U_3,\ldots, U_n$ and integrate over $U_2$. In this integral, replace $U_2$ by $e^{\I H_2}$, incurring another factor of $e^{6rN^2}C_N$. The proof of the upper bound is completed by repeating this process $n$ times, replacing each $U_j$ by $e^{\I H_j}$. The lower bound is obtained similarly.
\end{proof}


\section{Some standard results about Gaussian measures}\label{gausssec}
In this section we will review --- mostly without proof --- some results about Gaussian measures that will be needed in the subsequent sections. A Gaussian measure on $\rr^n$ is a probability measure that has density proportional to $e^{-P(x)}$ with respect to Lebesgue measure, where $P$ is a polynomial of degree two. Not all polynomials of degree two correspond to Gaussian measures; $P$ should have the property that $e^{-P(x)}$ is integrable. A necessary and sufficient condition for this to happen is that for some positive constant $c$, 
\begin{align}\label{intcrit}
P(x)\ge c\|x\|^2 \text{ whenever $\|x\|$ is sufficiently large,}
\end{align}
where $\|x\|$ denotes the Euclidean norm of $x$. 

A different way to express the above criterion for integrability is as follows. Suppose that $P(x)$ is written as
\[
P(x) = x^TQ x + v^Tx + c\,,
\]
where $x^T$ denotes the transpose of a column vector $x\in \rr^n$, $Q$ is an $n\times n$ matrix, $v\in \rr^n$ and $c\in \rr$. Any polynomial of degree two can be written in the above format. Then the criterion \eqref{intcrit} is the same as saying that:
\begin{align}\label{intcrit2} 
\text{$Q$ is a positive definite matrix.}
\end{align}
We will refer to $x^TQ x$ as the quadratic component of $P(x)$ and $v^Tx$ as the linear component of~$P(x)$.

There is a standard way of writing the probability density function of a Gaussian measure on $\rr^n$, which is the following:
\begin{align}\label{gaussdensity}
\frac{1}{(2\pi)^{n/2}(\det \Sigma)^{1/2}}\exp\biggl(-\frac{1}{2}(x-\mu)^T \Sigma^{-1} (x-\mu)\biggr)\,.
\end{align}
Here $\Sigma$ is an $n\times n$ positive definite matrix, $\mu\in \rr^n$ and $\det \Sigma$ is the determinant of $\Sigma$. If the density is expressed as a function proportional to 
\begin{equation*}
\exp(-x^TQ x - v^Tx - c)\,,
\end{equation*}
then the relation between the pairs $(Q,v)$ and $(\Sigma, \mu)$ is easy to read off by equating terms in the exponent:
\begin{align}
\Sigma &= \frac{1}{2}Q^{-1}\,,\label{musigma} \\
\mu &= -\Sigma v\,.\label{musigma2}
\end{align}
Let $X=(X_1,\ldots, X_n)$ be a random vector with probability density \eqref{gaussdensity}. Then for each $1\le j\le n$, 
\[
\ee(X_j) = \mu_j\,,
\]
where $\mu_j$ is the $j^{\mathrm{th}}$ coordinate of $\mu$, and for each $1\le j,k\le n$,
\[
\cov(X_j, X_k) = \ee(X_jX_k)-\ee(X_j)\ee(X_k) = \sigma_{jk}\,,
\]
where $\sigma_{jk}$ is the $(j,k)^{\mathrm{th}}$ entry of the matrix $\Sigma$. For this reason, $\mu$ is called the mean vector of $X$ and $\Sigma$ is called the covariance matrix of $X$. If the mean vector is zero, the Gaussian measure is said to be centered. 

One significance of the above formulas is that a Gaussian probability measure is completely determined by its mean vector and covariance matrix. 

A property of Gaussian random vectors that will be important for us is that if $X$ is a Gaussian random vector as above, then for any $A\subseteq \{1,\ldots, n\}$, the $\rr^A$-valued random vector $X_A := (X_j)_{j\in A}$ also has a Gaussian distribution, with means and covariances inherited from $X$. That is, the mean vector of $X_A$ is $\mu_A := (\mu_j)_{j\in A}$ and the covariance matrix of $X_A$ is $\Sigma_A := (\sigma_{jk})_{j,k\in A}$. 

Let $a$ and $b$ denote the minimum and maximum eigenvalues of $\Sigma$. Then recall that
\[
a = \min_{\|x\|=1} x^T \Sigma x\,, \ \  b = \max_{\|x\|=1} x^T \Sigma x\,.
\]
From this representation it follows easily that if $a_A$ and $b_A$ are the minimum and maximum eigenvalues of $\Sigma_A$, then 
\begin{equation}\label{aabb}
a\le a_A\le b_A\le b\,.
\end{equation}
Recall also that the eigenvalues of $\Sigma_A^{-1}$ are precisely the inverses of the eigenvalues of $\Sigma_A$, and that the determinant of $\Sigma_A$ is the product of the eigenvalues of $\Sigma_A$. Therefore it follows from the formula~\eqref{gaussdensity} and the above inequalities that the probability density of $X_A$ at a point $x_A\in \rr^A$ is bounded below by
\begin{equation}\label{cruciallower}
\frac{1}{(2\pi)^{|A|/2} b^{|A|/2}} \exp\biggl(-\frac{1}{2a}\|x_A-\mu_A\|^2\biggr)\,.
\end{equation}
One last fact about Gaussian random variables that we will need is the following well-known inequality. Suppose that $X_1,\ldots, X_n$ are as above. Clearly, the variance of each $X_j$ is bounded above by~$b$. Consequently, for any $x\ge 0$ and $\theta \ge 0$,
\begin{align*}
\pp(\max_{1\le j\le n} |X_j-\mu_j|\ge x) &\le \sum_{j=1}^n \pp(|X_j-\mu_j|\ge x)= 2\sum_{j=1}^n \pp(X_j-\mu_j\ge x)\\
&\le 2e^{-\theta x}\sum_{j=1}^n\ee(e^{\theta (X_j-\mu_j)}) \le 2n e^{-\theta x+\theta^2b/2}\,.
\end{align*}
Choosing $\theta = x/b$ gives the bound
\begin{align*}
\pp(\max_{1\le j\le n} |X_j-\mu_j|\ge x) &\le  2n e^{-x^2/2b}\,.
\end{align*}
Consequently, if $n\ge 2$ then 
\begin{equation}\label{gaussmax}
\pp\biggl(\max_{1\le j\le n} |X_j|\ge \max_{1\le j\le n} |\mu_j| + \sqrt{6b\log n}\biggr)\le \frac{2}{n^2} \le \frac{1}{2}\,.
\end{equation}
We will have crucial uses of \eqref{cruciallower} and \eqref{gaussmax} later. 


\section{Lattice Maxwell theory}\label{maxwellsec}
In this section, we define lattice Maxwell theory. This is a Gaussian theory that will eventually be used to approximate $U(N)$ lattice gauge theory in the $\beta \ra \infty$ limit. Lattice Maxwell theory puts a scalar variable on each edge of the lattice. We will later expand the theory to attached a skew-Hermitian matrix to each edge, representing the Lie algebra approximation of the Lie group $U(N)$ near the identity. For a discussion of lattice Maxwell theory as it arises in the classical literature, see Chapter 22 of \cite{gj87}. 

 Let $T(B_n) = \rr^{E_n}$ be the set of all real-valued functions on $E_n$. Recall the quadratic form $M_n$ on $T(B_n)$ that was defined in Section \ref{results}.  Take any set $E$ such that
\[
E_n^0\subseteq E \subseteq E_n\,,
\]
and let $T_E(B_n) := \rr^{E_n\backslash E}$. 
Let $\theta$ be a real-valued function on $E$. Take any $t\in T_E(B_n)$. Extend $t$ to an element $s\in T(B_n)$ by defining 
\[
s(x,y):=
\begin{cases}
\theta(x,y) &\text{ if } (x,y)\in E,\\
t(x,y) &\text{ if } (x,y)\in E_n\backslash E.
\end{cases}
\]
With $t$ extended to $s$ as above, define
\[
M_{E,\theta,n}(t) := M_n(s)\,.
\]
The space $T_E(B_n)$ is endowed with the natural Euclidean norm:
\[
\|t\|^2 := \sum_{(x,y)\in E_n\backslash E} t(x,y)^2\,.
\]
The function $\theta$ that is zero everywhere on $E$ is particularly important. This function will be denoted by the symbol~$0$. The following lemma shows that $M_{E,0,n}$ is a positive definite quadratic form on the vector space $T_{E}(B_n)$ and gives a lower bound for its smallest eigenvalue and an upper bound for its largest eigenvalue. One can say that this result is a kind of discrete Poincar\'e inequality.
\begin{lmm}\label{matrixlmm}
For each $t\in T_E(B_n)$, 
\[
\frac{C_1}{n^{d+2}}\|t\|^2\le M_{E,0,n}(t) \le C_2\|t\|^2\,,
\]
where $C_1$  and $C_2$ are positive constants that depend only on $d$. 
\end{lmm}
\noindent{\it Remark.} One can show that the optimal lower bound is $C_1 n^{-d}\|t\|^2$, but the cruder bound displayed above will suffice for the applications of this lemma in this  manuscript.
\begin{proof}
First, extend $t$ to an element $s\in T(B_n)$ by defining $s(x,y)=0$ for $(x,y)\in E$ and $s(x,y)=t(x,y)$ for $(x,y)\in E_n\backslash E$. Then $M_{E,0,n}(t) = M_n(s)$. The upper bound follows easily from the definition of $M_n(s)$, since each edge belongs to at most $C$ plaquettes, where $C$ depends only on $d$. For the lower bound, it suffices to prove by induction that for each $(x,y)\in E_n$, 
\begin{equation}\label{toprove}
|s(x,y)|\le |x|_1 \sqrt{M_n(s)}\,,
\end{equation}
because $|x|_1\le dn$ for every $x\in B_n$. (Recall that $|x|_1$ denotes the $\ell^1$ norm of $x$.)

We will prove \eqref{toprove} by induction on $|x|_1$. This is clearly true if $|x|_1=0$, since every edge incident to the origin belongs to $E_n^0$, and $E_n^0\subseteq E$. So assume that $|x|_1> 0$. Let $x_1,\ldots,x_d$ be the coordinates of $x$. Then $y = x+ e_j$ for some $j$. Let $k$ be the largest index such that $x_k\ne 0$. If $k\le j$ then $(x,y)\in E_n^0\subseteq E$, and therefore $s(x,y)=0$. So assume that $k >j$. Let 
\[
z := x-  e_k\,.
\] 
Then $(z,j,k)\in B_n'$. 
Now note that the edges $(z, z+ e_k)$ and $(z+ e_j, z +  e_j + e_k)$ belong to $E_n^0$. Therefore,
\begin{align*}
s(z,j,k) &= s(z, z+ e_j) + s(z+ e_j +  e_k, z+ e_k)\\
&= s(z, z+ e_j) - s(x,y)\,.
\end{align*}
This can be rewritten as
\[
s(x,y) = s(z,z+ e_j) - s(z,j,k)\,.
\]
If $(z,z+ e_j)\in E$, then $s(z,z+ e_j)=0$ and the above identity gives
\[
|s(x,y)|= |s(z,j,k)|\le \sqrt{M_n(s)}\le |x|_1\sqrt{M_n(s)}\,.
\]
If $(z,z+ e_j)\not \in E$, then since $|z|_1=|x|_1-1$, the induction hypothesis implies that
\[
|s(z,z+ e_j)|\le |z|_1 \sqrt{M_n(s)} = (|x|_1-1)\sqrt{M_n(s)}\,,
\]
and therefore
\begin{align*}
|s(x,y)| &\le |s(z,z+ e_j)|+|s(z,j,k)|\\
&\le (|x|_1-1) \sqrt{M_n(s)} + \sqrt{M_n(s)}= |x|_1\sqrt{M_n(s)}\,.
\end{align*}
This completes the induction step.
\end{proof}
By the lower bound from the above lemma and the criterion \eqref{intcrit2} discussed in Section \ref{gausssec}, the quadratic form $M_{E,0,n}$ defines a Gaussian measure on $T_E(B_n)$. We will denote this measure by $\tau_{E,0,n}$. Note that for any $\theta$, $M_{E,\theta,n}$ and $M_{E,0,n}$ differ by a linear function. In other words, the quadratic component of $M_{E,\theta,n}$ is the same as that of $M_{E,0,n}$. Thus, $M_{E,\theta,n}$ also defines a Gaussian measure on $T_E(B_n)$, which we will denote by $\tau_{E,\theta,n}$. This Gaussian measure will be called lattice Maxwell theory on $B_n$ with boundary value $\theta$. When $E = E_n^0$ and $\theta=0$, we will simply write $\tau_n$ instead of $\tau_{E_n^0,0,n}$ and $T_0(B_n)$ instead of $T_{E_n^0}(B_n)$.


\section{Some estimates for lattice Maxwell theory}\label{maxestsec}
In this section, we will prove three important small ball probability estimates for lattice Maxwell theory. Throughout, $C, C_1, C_2,\ldots$ will denote positive constants that depend only on $d$, whose values may change from line to line.  The first estimate is given by the following theorem.
\begin{thm}\label{latticethm1}
Take any $n\ge 2$. For any nonempty $A\subseteq E_n^1$ and $\eta\in (0,1/2]$,
\[
\tau_{n}(\{t\in T_0(B_n): |t(x,y)|\le \eta \textup{ for all } (x,y)\in A\}) \ge e^{-C|A|(\log(1/\eta)+\log n)}\,,
\]
where $C$ is a positive constant that depends only on $d$.
\end{thm}
\begin{proof}
Recall that $\tau_{n}$ is a centered Gaussian measure. Let $\Sigma$ denote the covariance matrix of this measure. Let $a$ and $b$ be the smallest and largest eigenvalues of $\Sigma$. By the relation \eqref{musigma} and Lemma \ref{matrixlmm}, 
\begin{equation}\label{smalllarge}
C_1 \le a \le b\le C_2n^{d+2}\,.
\end{equation}
Let $t$ be a random vector drawn from the measure $\tau_{n}$.  Then by the bounds from \eqref{smalllarge} and the lower bound \eqref{cruciallower}, it follows that 
\begin{align*}
\tau_n(\{t\in T_0(B_n): |t(x,y)|\le \eta \text{ for all } (x,y)\in A\}) &\ge\frac{ (2\eta)^{|A|}  e^{-|A|\eta^2/(2C_1)}}{(2\pi)^{|A|/2}(C_2 n^{d+2})^{|A|/2}}\,.
\end{align*}
Since $\eta\le 1/2$, the right side is bounded below by $ e^{-C|A|(\log (1/\eta) +\log n)}.$
\end{proof}

Next, take any $E_n^0\subseteq E \subseteq E_n$, and a function $\theta:E \ra \rr$. Let
\[
\|\theta\|^2 := \sum_{(x,y)\in E} \theta(x,y)^2\,.
\]
Our second goal in this section is to prove the following result.
\begin{thm}\label{latticethm2}
There exists a constant $C$ depending only on $d$ such that for any $E$, $\theta$ and $n$,
\begin{align*}
\tau_{E,\theta, n} (\{t\in T_E(B_n): |t(x,y)|\le Cn^{d+2}(1 + \|\theta\|)\textup{ for all } (x,y)\in E_n\})\ge\frac{1}{2}\,.
\end{align*}
\end{thm}
\begin{proof}
As noted before, $M_{E,\theta,n}$ and $M_{E,0,n}$ have the same quadratic component. This implies that the covariance matrix of the Gaussian measure $\tau_{E, \theta, n}$ is the same as that of $\tau_{E, 0, n}$. However, the mean vector may be different for the two measures. Fix some $\theta$ and let $\mu$ denote the mean vector of $\tau_{E,\theta, n}$. It is easy to see that the sum of squares of the coefficients in the linear component of $M_{E,\theta,n}$ is bounded above by $C_1\|\theta\|^2$, since each edge can belong to at most $C_2$ plaquettes. Thus, by the representation~\eqref{musigma2} and the upper bound from~\eqref{smalllarge}, 
\[
\|\mu\|\le C n^{d+2}\|\theta\|\,.
\] 
In particular, the absolute value of each component of $\mu$ is bounded by $C n^{d+2}\|\theta\|$. The result now follows easily by~\eqref{gaussmax}, using the upper bound from \eqref{smalllarge}.
\end{proof}
Theorem \ref{latticethm1} and Theorem \ref{latticethm2} help in proving the following theorem, which is one of the main steps in the proof of Theorem \ref{mainthm}. 
\begin{thm}\label{latticemain}
There exist positive constants $C_1$ and $C_2$, depending only on $d$, such that the following is true. Take any $n\ge m\ge 2$ such that $m\le \sqrt{n}$. Then 
\[
\tau_{n}(\{t\in T_0(B_n): |t(x,y)|\le C_1 m^{d+2} \textup{ for all } (x,y)\in E_n\}) \ge \exp\biggl(-\frac{C_2n^d\log n}{m}\biggr)\,.
\]
\end{thm}
\begin{proof}
Let $r$ be largest integer such that $r(m-1)\le n-1$. Let $\mb$ be the collection of subsets of $B_n$ of the form 
\begin{equation}\label{btrans}
\{(x_1,\ldots, x_d): a_j (m-1)\le x_j \le (a_j+1)(m-1), \, j=1,\ldots,d\}
\end{equation}
where $a_1,\ldots, a_d\in \{0,1,\ldots,r-1\}$. In other words, each element of $\mb$ is a translate of the box $B_{m}$, and two elements can intersect only at a common boundary.  Let $A$ be the collection of all elements of $E_n$ that belong to the boundaries of these boxes, plus all edges that are not contained in any of the boxes. It is easy to see that there are $O(n^{d-1}m)$ edges that are not contained in any element of $\mb$ and $O(n^d m^{-1})$ edges that belong to the boundaries of elements of $\mb$. Using the assumption that $m^2\le n$, this implies  
\begin{equation}\label{abbds}
|A|\le \frac{Cn^d}{m} + C n^{d-1} m\le \frac{Cn^d}{m}\,.
\end{equation}
Also, we have
\begin{equation}\label{abbds2}
|\mb|\le \frac{Cn^d}{m^d}\,.
\end{equation}
Let $t$ be a random configuration drawn from the measure $\tau_{n}$. Let
\[
t_A := (t(x,y))_{(x,y)\in A}\,,
\]
and for each $B\in \mb$, let 
\[
t_B:= (t(x,y))_{(x,y) \in E(B)\backslash A}\,,
\]
where $E(B)$ is the set of positively oriented edges of $B$. 
Consider the conditional distribution of $t$ given its values in $A$. It is easy to see from the definition of $\tau_{n}$ that under this conditioning, the $t_B$'s become independent random configurations, and each of them follows some $\tau_{E, \theta, m}$ distribution, where $E$ is determined by the location of the box $B$ and $\theta$ is determined by the values of $t$ on $A\cap E(B)$. Therefore by Theorem \ref{latticethm2}, there is a constant $C$ depending only on $d$ such that the conditional probability given $t_A$ of the event 
\begin{equation*}
\{|t(x,y)|\le C m^{d+2}(1+\|t_A\|) \text{ for all } (x,y)\in E(B)\backslash A\}
\end{equation*}
is at least $1/2$. Consequently, the conditional probability that these events happen simultaneously for all $B\in \mb$ is at least $(1/2)^{|\mb|}$. On the other hand, by Theorem \ref{latticethm1}, the event 
\begin{equation*}
\{|t(x,y)|\le n^{-d/2} \text{ for all } (x,y)\in A\}
\end{equation*}
has probability at least $e^{-C|A|\log n}$. If this event happens, then $\|t_A\|\le C$. The assertion of the theorem follows by combining the two lower bounds obtained above and using \eqref{abbds} and \eqref{abbds2}. 
\end{proof}


\section{The infinite volume limit of lattice Maxwell theory}\label{maxwelllim}
As in the previous section, $C, C_1, C_2,\ldots$ will denote positive constants that depend only on $d$, whose values may change from line to line. 
Let $l^0_{n}$ be the Lebesgue measure on $T_0(B_n)$. The partition function of lattice Maxwell theory on $B_n$ is defined as
\[
Z_M(B_n) := \int_{T_0(B_n)}e^{-\frac{1}{2}M_n(t)}\, dl^0_{n}(t)\,.
\]
The corresponding free energy per site is defined as
\[
F_M(B_n) := \frac{\log Z_M(B_n)}{n^d}\,.
\]
By the formula \eqref{gaussdensity} for the probability density of Gaussian measures, we see that
\begin{equation}\label{fmkn}
F_M(B_n) = K_{n,d} + \frac{|E_n^1|}{2n^d}\log (2\pi)\,,
\end{equation}
where $K_{n,d}$ is the constant defined in equation \eqref{kndeq} of Section \ref{results}. The following lemma gives an a priori bound on the magnitude of $F_M(B_n)$.
\begin{lmm}\label{fmbd0}
There is a constant $C$ depending only on $d$ such that for any $n$,
\[
|F_M(B_n)|\le C\log n\,.
\]
\end{lmm}
\begin{proof}
Recalling the formula \eqref{kndeq} for $K_{n,d}$, we need to only show that 
\[
|\log \det M_n^0|\le Cn^d \log n\,.
\]
By Lemma \ref{matrixlmm}, the smallest eigenvalue of $M_n^0$ is at least $C_1 n^{-(d+2)}$ and the largest eigenvalue is at most $C_2$. Therefore
\[
\frac{C_1}{n^{(d+2)n^d}} \le \det M_n^0 \le C_2^{n^d}\,.
\]
The proof is completed by taking logarithms on both sides. 
\end{proof}
The following theorem establishes the existence of the infinite volume limit of the free energy per site of lattice Maxwell theory. This is the main result of this section.
\begin{thm}\label{cdthm}
As $n\ra\infty$, $F_M(B_n)$ converges to a finite limit.
\end{thm}
The proof of Theorem \ref{cdthm} requires two lemmas. First, take any $n\ge 3$. Let $E_n'$ denote the union of  $E_n^0$ and the set of boundary edges of $B_n$. The set $E_n'$ in dimension two is depicted in Figure \ref{enpfig}. 

\begin{figure}[t]
\begin{pspicture}(1,.5)(11,7.5)
\psset{xunit=1cm,yunit=1cm}
\psline[linestyle = dashed]{-}(3,1)(9,1)
\psline{-}(3,2)(9,2)
\psline{-}(3,3)(9,3)
\psline{-}(3,4)(9,4)
\psline{-}(3,5)(9,5)
\psline{-}(3,6)(9,6)
\psline[linestyle = dashed]{-}(3,7)(9,7)
\psline[linestyle = dashed]{-}(3,1)(3,7)
\psline[linestyle = dashed]{-}(4,1)(4,7)
\psline[linestyle = dashed]{-}(5,1)(5,7)
\psline[linestyle = dashed]{-}(6,1)(6,7)
\psline[linestyle = dashed]{-}(7,1)(7,7)
\psline[linestyle = dashed]{-}(8,1)(8,7)
\psline[linestyle = dashed]{-}(9,1)(9,7)
\end{pspicture}
\caption{The dashed lines represent the edges belonging to $E_n'$, for $n=7$ and $d=2$. Note the slight difference with the edge set $E_n^0$ shown in Figure \ref{en0fig}.}
\label{enpfig}
\end{figure}
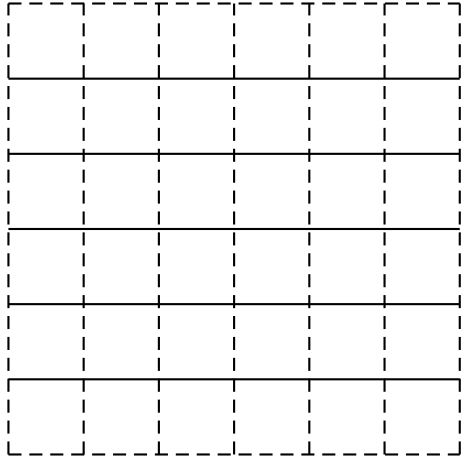

For simplicity, let us write $T'(B_n)$ instead of $T_{E_n'}(B_n)$. Let $l'_{n}$ be the Lebesgue measure on $T'(B_n)$. Recall the definition of the quadratic form $M_{E_n', 0,n}$ on $T'(B_n)$ from Section \ref{maxwellsec}. Let us denote it simply by $M_n$. Let $D_n$ be the set of all $t\in T'(B_n)$ such that $|t(x,y)|\le n^{-(d-1)/2}$ for any $(x,y)$ that belongs to a plaquette that touches the boundary of $B_n$. Define
\[
Z_M'(B_n) := \int_{D_n} e^{-\frac{1}{2}M_n(t)}\, dl'_n(t)
\]
and
\[
F_M'(B_n) := \frac{1}{n^d}\log Z_M'(B_n)\,.
\]
The following lemma is the first ingredient in the proof of Theorem \ref{cdthm}.
\begin{lmm}\label{cdlmm1}
There is a constant $C$ depending only on $d$ such that for any $n\ge 3$,
\[
|F_M'(B_n)-F_M(B_n)|\le \frac{C\log n}{n}\,.
\]
\end{lmm}
\begin{proof}
Let $A_n$ be the subset of $T_0(B_n)$ consisting of all $t$ such that $|t(x,y)|\le n^{-(d-1)/2}$ for any $(x,y)$ that belongs to a plaquette that touches the boundary of $B_n$. By Theorem \ref{latticethm1},
\begin{align}\label{zmbn1}
e^{-Cn^{d-1}\log n}Z_M(B_n)\le \int_{A_n} e^{-\frac{1}{2}M_n(t)}\, dl^0_{n}(t) \le Z_M(B_n)\,.
\end{align}
Take any $t\in T_0(B_n)$. Let $t'$ be the configuration obtained by changing the value of $t$ everywhere on $E'_n$ to zero. Define
\[
M_{n}'(t) := M_n(t')\,.
\]
Note that by the defining property of $A_n$, it is easy to see that for any $t\in A_n$,
\[
|M_n(t) - M_{n}'(t)|\le C\,.
\]
Thus, 
\begin{align}\label{zmbn2}
e^{-C}\int_{A_n} e^{-\frac{1}{2}M_{n}'(t)}\, dl^0_{n}(t)\le \int_{A_n} e^{-\frac{1}{2}M_n(t)}\, dl^0_{n}(t)\le e^C\int_{A_n} e^{-\frac{1}{2}M_{n}'(t)} \,dl^0_{n}(t)\,.
\end{align}
Now note that by the definition of $M_{n}'$, $M_{n}'(t)$ has no dependence on the values of $t$ at the edges that belong to $E_n'$. Therefore by the product nature of $l^0_{n}$ and the rectangular nature of $A_n$, we can integrate out the edges in $E_n' \backslash E_n^0$ and get
\begin{align}\label{zmbn3}
\int_{A_n} e^{-\frac{1}{2}M_{n}'(t)}\, dl^0_{n}(t) = (2n^{-(d-1)/2})^{|E_n'\backslash E_n^0|}\int_{D_n} e^{-\frac{1}{2}M_{n}(t)} \,dl'_{n}(t)\,.
\end{align}
The proof is now easily completed by combining \eqref{zmbn1}, \eqref{zmbn2} and \eqref{zmbn3}, and the fact that $|E_n'\backslash E_n^0| = O(n^{d-1})$.
\end{proof}
Next, take $3\le m\le n$ such that $n-1$ is a multiple of $m-1$. Let $\mb$ be as in the the proof of Theorem~\ref{latticemain}. Each element of $\mb$ is a translate of $B_{m}$. Therefore, the set of edges of an element of $\mb$ contains a translate of $E_{m}'$, where $E_m'$ is the union of $E_m^0$ and the boundary edges of $B_m$, as defined above.  Let $E_n^{m}$ be the union of all these translates. The set $E_n^m$ in dimension two is depicted in Figure \ref{emnfig}. 

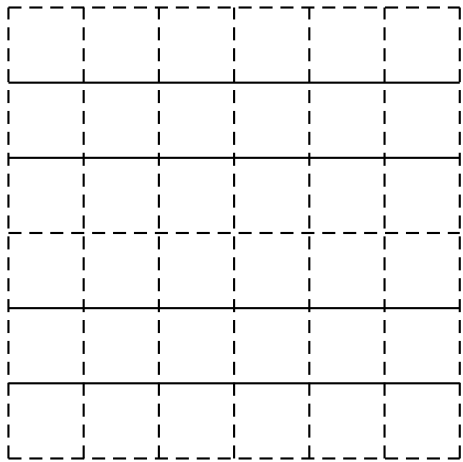
\begin{figure}[t]
\begin{pspicture}(1,.5)(11,7.5)
\psset{xunit=1cm,yunit=1cm}
\psline[linestyle = dashed]{-}(3,1)(9,1)
\psline{-}(3,2)(9,2)
\psline{-}(3,3)(9,3)
\psline[linestyle = dashed]{-}(3,4)(9,4)
\psline{-}(3,5)(9,5)
\psline{-}(3,6)(9,6)
\psline[linestyle = dashed]{-}(3,7)(9,7)
\psline[linestyle = dashed]{-}(3,1)(3,7)
\psline[linestyle = dashed]{-}(4,1)(4,7)
\psline[linestyle = dashed]{-}(5,1)(5,7)
\psline[linestyle = dashed]{-}(6,1)(6,7)
\psline[linestyle = dashed]{-}(7,1)(7,7)
\psline[linestyle = dashed]{-}(8,1)(8,7)
\psline[linestyle = dashed]{-}(9,1)(9,7)
\end{pspicture}
\caption{The dashed lines represent the edges belonging to $E_n^m$, for $n=7$, $m= 4$ and $d=2$.}
\label{emnfig}
\end{figure}

For simplicity, let us write $T_{m}(B_n)$ instead of $T_{E_n^{m}}(B_n)$. Let $l^{m}_{n}$ be the Lebesgue measure on $T_{m}(B_n)$. Recall the definition of the quadratic form $M_{E_n^m, 0,n}$ on $T_m(B_n)$ from Section \ref{maxwellsec}. As before, let us denote it simply by $M_n$. Let $S_{m,n}$ be set of all $t\in T_m(B_n)$ such that $|t(x,y)|\le m^{-(d-1)/2}$ for any $(x,y)$ that belongs to a plaquette that touches the boundary of any element of $\mb$. Define
\[
Z_{M,m}(B_n) := \int_{S_{m,n}} e^{-\frac{1}{2}M_{n}(t)}\, dl^m_n(t)
\]
and 
\[
F_{M,m}(B_n) := \frac{1}{n^d} \log Z_{M,m}(B_n)\,.
\]
The next lemma is the second component of the proof of Theorem \ref{cdthm}.
\begin{lmm}\label{cdlmm2}
For any $3\le m\le n$ such that $n-1$ is a multiple of $m-1$, 
\[
|F_{M,m}(B_n)-F_M(B_n)|\le \frac{C\log n}{m}\,,
\]
where $C$ depends only on $d$. 
\end{lmm}
\begin{proof}
Let $R_{m,n}$ be the subset of $T_0(B_n)$ consisting of all $t$ such that $|t(x,y)|\le m^{-(d-1)/2}$ for any $(x,y)$ that belongs to a plaquette that touches the boundary of any element of $\mb$ to which $(x,y)$ belongs. By Theorem \ref{latticethm1},
\begin{align}\label{zmmbn1}
e^{-Cn^dm^{-1}\log n}Z_M(B_n)\le \int_{R_{m,n}} e^{-\frac{1}{2}M_n(t)} \,dl^0_{n}(t) \le Z_M(B_n)\,.
\end{align}
Take any $t\in T_0(B_n)$. Let $t^m$ be the configuration obtained by changing the value of $t$ everywhere on $E^m_n$ to zero. Define
\[
M_{n,m}(t) := M_n(t^m)\,.
\]
It is easy to see that for any $t\in R_{m,n}$, 
\[
|M_n(t) - M_{n,m}(t)|\le \frac{Cn^d}{m^{d}}\,.
\]
Thus, 
\begin{align}
e^{-Cn^dm^{-d}}\int_{R_{m,n}} e^{-\frac{1}{2}M_{n,m}(t)}\, dl^0_{n}(t)&\le \int_{R_{m,n}} e^{-\frac{1}{2}M_n(t)} \,dl^0_{n}(t)\nonumber \\
&\le e^{Cn^d m^{-d}}\int_{R_{m,n}} e^{-\frac{1}{2}M_{m,n}(t)}\, dl^0_{n}(t)\,. \label{zmmbn2}
\end{align}
Note that by the definition of $M_{n,m}$, $M_{n,m}(t)$ has no dependence on the values of $t$ at the edges that belong to $E_n^m$. Therefore by the product nature of $l^0_{n}$ and the rectangular nature of $R_{m,n}$, we can integrate out the edges in $E_n^m \backslash E_n^0$ and get
\begin{align}\label{zmmbn3}
\int_{R_{m,n}} e^{-\frac{1}{2}M_{n,m}(t)} \,dl^0_{n}(t) = (2m^{-(d-1)/2})^{|E_n^m\backslash E_n^0|}\int_{S_{m,n}} e^{-\frac{1}{2}M_{n}(t)} \,dl^m_{n}(t)\,.
\end{align}
The proof is now easily completed by combining \eqref{zmmbn1}, \eqref{zmmbn2} and \eqref{zmmbn3}, and the fact that $|E^m_n\backslash E_n^0| = O(n^d/m)$. 
\end{proof}
We are now ready to prove Theorem \ref{cdthm}.
\begin{proof}[Proof of Theorem \ref{cdthm}]
Take any $3\le m\le n$ such that $n-1$ is a multiple of $m-1$. It is easy to see that $Z_{M,m}(B_n)$ breaks up as a product of $|\mb|$ integrals, and each element of the product equals $Z_{M}'(B_m)$. In other words,
\[
Z_{M,m}(B_n) = (Z_{M}'(B_m))^{|\mb|}\,.
\]
Therefore,
\[
F_{M,m}(B_n)= \frac{|\mb|m^d}{n^d}F_M'(B_m)\,.
\]
Since
\[
|\mb| = \frac{n^d}{m^d} + O\biggl(\frac{n^d}{m^{d+1}}\biggr)\,,
\]
the identity from the previous display, Lemma \ref{fmbd0} and Lemma \ref{cdlmm1} imply that
\begin{align*}
|F_{M,m}(B_n)-F_M'(B_m)|&\le \frac{C}{m}|F_M'(B_m)|\\
&\le \frac{C}{m}\biggl(|F_M(B_m)|+\frac{C\log m}{m}\biggr)\\
&\le \frac{C\log m}{m}\,.
\end{align*}
Therefore by Lemma \ref{cdlmm2},
\begin{align*}
|F_{M}'(B_m) - F_M(B_n)|\le \frac{C\log n}{m}\,.
\end{align*}
By the above inequality and Lemma \ref{cdlmm1}, we get
\begin{align}\label{fmbn3}
|F_{M}'(B_m)-F_{M}'(B_n)|\le \frac{C\log n}{m}\,.
\end{align}
It follows from this inequality that the sequence $\{F_{M}'(B_{2^k+1})\}_{k\ge 1}$ is Cauchy and hence convergent. Now take any $l$ and suppose that $2^k+1\le l\le 2^{k+1}+1$. Let $m := 2^k+1$ and let
\[
n := (m-1)(l-1)+1\,.
\] 
Then both $m-1$ and $l-1$ divide $n-1$. Thus, by \eqref{fmbn3}, 
\begin{align*}
|F_{M}'(B_l)-F_{M}'(B_m)|&\le |F_{M}'(B_m)-F_{M}'(B_n)| + |F_{M}'(B_l)-F_{M}'(B_n)|\\
&\le \frac{C\log n}{m} + \frac{C\log n}{l}\le \frac{Ck}{2^k}\,.
\end{align*}
Thus, $F_{M}'(B_n)$ converges as $n\ra\infty$. By Lemma \ref{cdlmm1}, this implies that $F_M(B_n)$ is convergent.
\end{proof}


\section{From Wilson action to Maxwell action}\label{wmsec}
The goal of this section is to complete the program started in Section \ref{liesec}, by giving a concrete prescription, with error bounds, for replacing the Wilson action by the action of lattice Maxwell theory. This will help in replacing the integrals over $U(N)$ with integrals over $H(N)$. Recall that $H(N)$ is the set of $N\times N$ Hermitian matrices. Although the Lie algebra of $U(N)$ is the set of skew-Hermitian matrices, it suffices to work with Hermitian matrices since a skew-Hermitian matrix is just a Hermitian matrix multiplied by $\I$. 

Suppose that to each $(x,y)\in E_n$ we attach a Hermitian matrix $H(x,y)\in H(N)$, and let $H(y,x) := -H(x,y)$. Let $H(B_n)$ be the set of all such assignments. Given a configuration $H\in H(B_n)$ and a plaquette $(x,j,k)\in B_n'$, define
\begin{align*}
H(x,j,k) &:= H(x,x+ e_j) + H(x+ e_j, x+ e_j +  e_k)\\
&\qquad  + H(x+ e_j+ e_k, x+ e_j) + H(x+ e_j, x)\,.
\end{align*}
The Maxwell action for a configuration $H\in H(B_n)$ is defined as
\[
M_n(H) := \sum_{(x,j,k)\in B_n'} \|H(x,j,k)\|^2\,.
\]
Let $H_0(B_n)$ be the set of all $H\in H(B_n)$ such that $H(x,y)=0$ for all $(x,y)\in E_n^0$. Now recall the definition of Lebesgue measure on $H(N)$ that we adopted in Section \ref{liesec}. This naturally leads to a definition of a product Lebesgue measure $\lambda^0_{n}$ on $H_0(B_n)$. Define
\[
Z_H(B_n) := \int_{H_0(B_n)} e^{-\frac{1}{2}M_n(H)} \,d\lambda^0_{n}(H)
\]
and
\[
F_H(B_n) := \frac{\log Z_H(B_n)}{n^d}\,.
\]
From the way Lebesgue measure on $H(N)$ was defined in Section \ref{liesec}, and the way that $M_n(H)$ is defined here, it easy to see that the integral defining $Z_H(B_n)$ can be written as a product of $N^2$ integrals, giving
\begin{equation*}
Z_H(B_n) = (Z_M(B_n))^{N^2}\,,
\end{equation*}
where $Z_M(B_n)$ is the partition function of lattice Maxwell theory defined in Section \ref{maxwelllim}. Therefore
\begin{align}\label{zhzm}
F_H(B_n) = N^2 F_M(B_n)\,.
\end{align}
The following lemma gives a quadratic approximation of $e^{\I H}$ for a Hermitian matrix $H$. 
\begin{lmm}\label{huerr}
If $H$ is a Hermitian matrix, then
\[
\biggl\|e^{\I H} - I - \I H + \frac{H^2}{2}\biggr\|\le \frac{\|H\|^3}{6}\,.
\]
\end{lmm}
\begin{proof}
Define a matrix valued function 
\[
f(t) := e^{\I t H} - I - \I t H + \frac{t^2H^2}{2}\,.
\]
Then $f(0)=f'(0)=f''(0)=0$, and  by Lemma \ref{uhslmm},
\[
\|f'''(t)\| = \|H^3e^{\I tH}\|= \|H^3\|\le \|H\|^3\,.
\]
Thus,
\begin{align*}
\biggl\|e^{\I H} - I - \I H + \frac{H^2}{2}\biggr\| &= \|f(1)\|\\
&= \biggl\|\frac{1}{2}\int_0^1(1-t)^2f'''(t)\,dt\biggr\|\\
&\le \frac{1}{2}\int_0^1(1-t)^2\|f'''(t)\|\,dt\le \frac{\|H\|^3}{6}\,.
\end{align*}
This completes the proof of the lemma.
\end{proof}
For any $H\in H(B_n)$, let $U = e^{\I H}\in U(B_n)$  be the configuration that is defined as
\[
U(x,y) := e^{\I H(x,y)}\,.
\]
The next lemma extends the quadratic approximation of $e^{\I H}$ to configurations $H\in H(B_n)$. 
\begin{lmm}\label{utohlmm}
Take any $H\in H(B_n)$ and let $U := e^{\I H}$. Suppose that $\eta$ is a number such that for all $(x,y)\in E_n$, $\|H(x,y)\|\le \eta \le 1$. Then for any $(x,j,k)\in B_n'$, 
\[
\biggl|\phi(U(x,j,k)) - \frac{1}{2}\|H(x,j,k)\|^2\biggr|\le C\eta^3\,,
\]
where $C$ depends only on $N$. 
\end{lmm}
\begin{proof}
For each $(x,y)\in E_n$, let
\[
R(x,y) := I + \I H(x,y) - \frac{H(x,y)^2}{2}
\]
and let
\begin{align*}
R(y,x) &:= I + \I H(y,x) - \frac{H(y,x)^2}{2}\\
&= I- \I H(x,y) - \frac{H(x,y)^2}{2}\,.
\end{align*}
For $(x,j,k)\in B_n'$, let
\[
R(x,j,k) := R(x,x+ e_j)R(x+ e_j, x+ e_j +  e_k) R(x+ e_j+e_k, x+e_j) R(x+ e_j, x)\,.
\]
By Lemma \ref{huerr}, for any $(x,y)\in E_n$, $\|U(x,y)-R(x,y)\|$ and $\|U(y,x)-R(y,x)\|$ are both bounded by $\eta^3/6$. Since $\eta \le 1$, this implies in particular that $\|R(x,y)\|$ and $\|R(y,x)\|$ are bounded by $C$ for all $(x,y)$, where $C$ stands for a constant that depends only on $N$. From these two observations it follows easily that for each $(x,j,k)\in B_n'$, 
\begin{align*}
\|U(x,j,k)-R(x,j,k)\|&\le C\eta^3\,.
\end{align*}
By the triangle inequality, this gives
\[
|\|I-U(x,j,k)\| - \|I- R(x,j,k)\|| \le C\eta^3\,.
\]
Since $\|I-U(x,j,k)\|\le C$, this implies in particular that  $\|I-R(x,j,k)\|\le C$. Thus, by Lemma~\ref{phihslmm},
\begin{align}\label{phiur}
&|\phi(U(x,j,k)) - \phi(R(x,j,k))|= \frac{1}{2}|\|I-U(x,j,k)\|^2 - \|I-R(x,j,k)\|^2|\le C\eta^3\,.
\end{align}
Now take any $(x,j,k)\in B_n'$. For simplicity, let
\begin{align*}
H_1 &:= H(x,x+e_j)\,,\\
H_2 &:= H(x+ e_j,x+ e_j +  e_k)\,,\\
H_3 &:= H(x+e_j+ e_k, x+ e_k)\,,\\
H_4 &:= H(x+e_k, x)\,.
\end{align*}
With this notation, observe that
\begin{align*}
R(x,j,k) &= I + \I(H_1+H_2+H_3+H_4) - \frac{1}{2}(H_1^2+H_2^2+H_3^2+H_4^2) \\
&\qquad - (H_1H_2 + H_1H_3+H_1H_4+H_2H_3+H_2H_4+H_3H_4)\\
&\qquad + \text{cubic and higher order terms.}
\end{align*}
Since $H_j$'s are Hermitian matrices, $\tr(H_1+H_2+H_3+H_4)$ is real. Therefore
\[
\Re(\tr(\I(H_1+H_2+H_3+H_4)))=0\,.
\]
Since $\tr(AB)=\tr(BA)$,
\begin{align*}
&\frac{1}{2}\tr(H_1^2+H_2^2+H_3^2+H_4^2) + \tr(H_1H_2 + H_1H_3+H_1H_4+H_2H_3+H_2H_4+H_3H_4)\\
&= \frac{1}{2}\tr((H_1+H_2+H_3+H_4)^2) = \frac{1}{2}\|H(x,j,k)\|^2\,.
\end{align*}
Lastly, observe that the absolute values of the traces of the cubic terms are bounded above by $C\eta^3$. Combining these observations, we get
\[
\phi(R(x,j,k)) = \Re(\tr(I-R(x,j,k))) = \frac{1}{2}\|H(x,j,k)\|^2 + \text{remainder}\,,
\]
where the absolute value of the remainder term is bounded by $C\eta^3$. The proof is completed by combining this  with \eqref{phiur}. 
\end{proof}
A direct consequence of Lemma \ref{utohlmm} is the following theorem, which is the main result of this section.
\begin{thm}\label{utohthm}
Take any $H\in H(B_n)$. Suppose that $r$ is a number such that $\|H(x,y)\|\le r$ for every $(x,y)\in E_n$. Then
\[
\biggl|S_{B_n}(e^{\I H}) - \frac{1}{2}M_n(H)\biggr|\le C r^3 n^d\,,
\]
where $C$ depends only on $d$ and $N$. 
\end{thm}
\begin{proof}
Let $U:= e^{\I H}$. Recalling the formulas for $S_{B_n}(e^{\I H})$ and $M_n(H)$, we see that
\[
S_{B_n}(e^{\I H}) - \frac{1}{2}M_n(H) =  \sum_{(x,j,k)\in B_n'} \biggl(\phi(U(x,j,k)) - \frac{1}{2}\|H(x,j,k)\|^2\biggr)\,.
\]
The proof is now easily completed by applying Lemma \ref{utohlmm}.
\end{proof}


\section{Proof of the main theorem}\label{proofmain}
The goal of this section is to wrap up the proof of Theorem \ref{mainthm} by connecting the various threads derived in the previous sections. Throughout this section, $C_N$ is the constant defined in the statement of Theorem \ref{liethm}. Let $Z_M$ and $F_M$ be defined as in Section \ref{maxwelllim} and $Z_H$ and $F_H$ be defined as in Section \ref{wmsec}. 

The proof is divided into two parts. First, we will prove an upper bound on the $\limsup$ of the free energy per site that agrees with the formula given in Theorem \ref{mainthm}, and then we will establish the matching lower bound on the $\liminf$ of the free energy per site. Establishing the upper bound requires a few lemmas. The first lemma, stated below, calculates the size of $E_n^1$.
\begin{lmm}\label{en1lmm}
For each $n\ge 2$,
\[
|E_n^1| = (d-1)n^d -d n^{d-1}+1\,.
\]
\end{lmm}
\begin{proof}
For $0\le l\le d$, let $B_n^l$ be the set of all $x\in B_n$ such that $l$ of the coordinates of $x$ that are less than $n-1$ and the rest are equal to $n-1$. Then there are exactly $l$ choices of $y$ such that $(x,y)\in E_n$. Thus, 
\[
|E_n| = \sum_{l=1}^d l |B_n^l|\,.
\]
Now, clearly,
\begin{align*}
|B_n^l| = {d \choose l} (n-1)^l\,.
\end{align*}
Therefore
\begin{align*}
|E_n| &= \sum_{l=1}^d {d\choose l} l(n-1)^l = d n^{d-1}(n-1)\,.
\end{align*}
Next, take any $(x,y)\in E_n^0$. Then there exists $1\le j\le d$, $x_1,\ldots,x_{j-1}\in \{0,1,\ldots, n-1\}$ and $x_j\in \{0,1,\ldots, n-2\}$ such that
\[
 x= (x_1,\ldots, x_{j-1}, x_j, 0,\ldots,0)
\]
and
\[
y = (x_1,\ldots, x_{j-1}, x_j+1,0,\ldots,0)\,.
\]
For each choice of $j$, $x_1,\ldots, x_{j-1}$ can be chosen in $n^{j-1}$ ways and $x_j$ can be chosen in $(n-1)$ ways. Therefore
\begin{align*}
|E_n^0| &= \sum_{j=1}^d n^{j-1}(n-1) = n^d-1\,.
\end{align*}
This completes the proof, since $|E_n^1|=|E_n|-|E_n^0|$. 
\end{proof}
The next lemma is an extension of Theorem \ref{mainstep1}, which states that a desired upper bound holds if $n$ is not too large (depending on $\beta$). 
\begin{lmm}\label{uplmm1}
Define 
\[
\delta_{n,\beta} := n^{(d+1)/2}\biggl(\frac{\log \beta}{\beta}\biggr)^{1/2}\,.
\]
Then for any $n$ and $g_0$, 
\begin{align*}
F(B_n, g_0)&\le  \frac{|E_n^1|}{n^d}\log C_N-\frac{|E_n^1|}{2n^d}N^2\log \beta + N^2F_M(B_n) + \frac{C}{n^d}+ C\delta_{n,\beta} + C\beta \delta_{n,\beta}^3\,,
\end{align*}
where $C$ depends only on $N$ and $d$.
\end{lmm}
\begin{proof}
Take any $n\ge 2$. Let $C_1$ be the constant from Theorem \ref{mainstep1}. Let $r_0$ be as in Theorem \ref{liethm}. First, suppose that 
\begin{equation}\label{rassume}
r:= \frac{1}{2}C_1\delta_{n,\beta}\le r_0\,.
\end{equation}
Let $U_0^\beta(B_n)$ be as in Theorem \ref{mainstep1}. Let $B(I,r)$ and $b(0,r)$ be as in Section \ref{liesec}, so that if $U\in U_0^\beta(B_n)$, then $U(x,y)\in B(I, 2r)$ for each $(x,y)\in E_n^1$. Let $T$ be the subset of $H_0(B_n)$ where $H(x,y)\in b(0,3r)$ for each $(x,y)\in E_n^1$. Then by the upper bound from Theorem \ref{liethm}, 
\begin{align*}
\int_{U_0^\beta(B_n)}  e^{-\beta S_{B_n}(U)}\, d\sigma^0_{B_n}(U)&\le C_N^{|E_n^1|}e^{Crn^d}\int_{T}e^{-\beta S_{B_n}(e^{\I H})}\, d\lambda_n^0(H)\,.
\end{align*}
But by Theorem \ref{utohthm}, for any $H\in T$,
\[
\biggl|S_{B_n}(e^{\I H}) - \frac{1}{2}M_n(H)\biggr|\le Cr^3n^d\,.
\]
Thus,
\begin{align*}
\int_{U_0^\beta(B_n)}  e^{-\beta S_{B_n}(U)}\, d\sigma^0_{B_n}(U) &\le C_N^{|E_n^1|}e^{C(r+\beta r^3 )n^d}\int_{H_0(B_n)}e^{-\frac{1}{2}\beta M_n(H)}\, d\lambda_{n}^0(H)\,.
\end{align*}
Making the change of variable $G = \sqrt{\beta} H$, we get
\begin{align*}
\int_{H_0(B_n)}e^{-\frac{1}{2}\beta M_n(H)}\, d\lambda_{n}^0(H) &= \beta^{-\frac{1}{2}N^2 |E_n^1|}Z_H(B_n)\,.
\end{align*}
Combining the steps and applying Theorem~\ref{mainstep1}, we get
\begin{align*}
F(B_n, g_0) &\le \frac{\log 2}{n^d} + \frac{|E_n^1|}{n^d}\log C_N-\frac{|E_n^1|}{2n^d}N^2\log \beta + F_H(B_n)  + Cr + C\beta r^3\,.
\end{align*}
The proof is now completed by applying the identity \eqref{zhzm}. 

Note that the above proof was executed under the assumption \eqref{rassume}. Suppose now that this assumption is violated. The assertion of the lemma is trivially true in this case, for a large enough choice of $C$. This is because $F(B_n, g_0)\le 0$, $|E_n^1|/n^d$ is uniformly bounded irrespective of $n$ (Lemma~\ref{en1lmm}), and $F_M(B_n)$ is also uniformly bounded irrespective of $n$ (Theorem \ref{cdthm}).
\end{proof}
The next lemma is a key tool in removing the condition about the smallness of $n$ that is present in Lemma \ref{uplmm1}. 
\begin{lmm}\label{uplmm2}
For any $2\le m\le n$,
\[
F(B_n, g_0)\le \biggl(1-\frac{Cm}{n}\biggr)F(B_m, g_0)\,,
\]
where $C$ depends only on $d$.
\end{lmm}
\begin{proof}
Suppose that $n = km + r$, where $0\le r\le m-1$. Then $B_{km}$ is contained in $B_n$ and is a disjoint union of $k^d$ translates of $B_m$. Since $\phi$ is a nonnegative function, this implies that
\[
Z(B_n, g_0) \le Z(B_{km}, g_0)\le (Z(B_m, g_0))^{k^d}\,.
\]
Take logarithm on both sides and dividing by $n^d$, we get
\[
F(B_n, g_0) \le \frac{k^dm^d}{n^d} F(B_m, g_0) = \frac{(n-r)^d}{n^d} F(B_m, g_0)\,.
\]
This completes the proof, since $r\le m-1$ and $F(B_m, g_0)\le 0$.
\end{proof}
Finally, the following lemma gives the required upper bound on the $\limsup$ of the free energy per site, without any restriction on the  growth rates of $n$ and $\beta$. 
\begin{lmm}\label{uplmmmain}
\[
\limsup_{\substack{n\ra\infty\\g_0\ra 0}} \biggl(F(B_n, g_0) + \frac{|E_n^1|}{2n^d}N^2\log \beta\biggr) \le (d-1)\log C_N+ N^2\lim_{n\ra\infty} F_M(B_n)\,.
\]
\end{lmm}
\begin{proof}
Let $\delta_{n,\beta}$ be as in Lemma \ref{uplmm1}. Suppose that $n\ra \infty$ and $g_0\ra 0$ simultaneously in such a way that both $\delta_{n,\beta}$ and $\beta \delta_{n,\beta}^3$ tend to zero. In this case, the result follows by Lemma~\ref{uplmm1}, Lemma~\ref{en1lmm} and Theorem~\ref{cdthm}. 

Next, suppose that at least one of $\delta_{n,\beta}$ and $\beta \delta_{n,\beta}^3$  remains bounded away from zero along a subsequence. We will show that the claimed inequality holds if the $\limsup$ is taken along that subsequence. This will complete the proof of the lemma. Without loss of generality, we may assume that the subsequence is the whole sequence. Then the assumed condition implies that $n$ must be growing at least as fast as a positive power of $\beta$. This allows us to define a third parameter, $m$, varying like a small positive power of $\beta$, so slowly that the following conditions hold:
\begin{align}\label{mcond1}
\delta_{m,\beta}\ra 0\,, \ \ \beta \delta_{m,\beta}^3 \ra 0 \,, \ \ \frac{m\log \beta}{n} \ra 0\,.
\end{align}
Note that the third condition in the above display implies that 
\begin{equation}\label{mcond2}
\frac{m}{n} \ra 0\,.
\end{equation}
By Lemma \ref{uplmm1} and Lemma \ref{uplmm2}, 
\begin{align*}
&F(B_n, g_0) + \frac{|E_n^1|}{2n^d}N^2\log \beta \le \biggl(1-\frac{Cm}{n}\biggr) F(B_m, g_0)+\frac{|E_n^1|}{2n^d}N^2\log \beta\\
&\le \biggl(1-\frac{Cm}{n}\biggr)\biggl(\frac{|E_m^1|}{m^d}\log C_N+ N^2F_M(B_m) +\frac{C}{m^d}+ C\delta_{m,\beta} + C\beta \delta_{m,\beta}^3\biggr)\\
&\qquad +\biggl(\frac{|E_n^1|}{2n^d} - \frac{|E_m^1|}{2m^d}\biggl(1-\frac{Cm}{n}\biggr)\biggr)N^2\log \beta \,.
\end{align*}
By \eqref{mcond1}, \eqref{mcond2}, Lemma \ref{en1lmm}, Theorem \ref{cdthm} and the fact that $m \ra \infty$, the first term tends to
\begin{align*}
(d-1)\log C_N + N^2\lim_{m\ra\infty} F_M(B_m)\,.
\end{align*}
By \eqref{mcond1}, Lemma \ref{en1lmm} and the fact that $m$ grows like a positive power of $\beta$,
\begin{align*}
\biggl(\frac{|E_n^1|}{2n^d} - &\frac{|E_m^1|}{2m^d}\biggl(1-\frac{Cm}{n}\biggr)\biggr)N^2\log \beta\le \frac{C\log \beta}{m}  + \frac{Cm\log \beta}{n} \ra 0\,.
\end{align*}
Combining the last three displays completes the proof.
\end{proof}
Next, we begin our quest for the lower bound. The following lemma shows that it suffices to work with $n$ that is not too large (depending on $\beta$).
\begin{lmm}\label{downlmm1}
For any $m$ and $n$,
\[
|F(B_m, g_0) - F(B_n, g_0)|\le \frac{C\beta}{m} + \frac{C\beta}{n}\,,
\]
where $C$ depends only on $N$ and $d$. 
\end{lmm}
\begin{proof}
Take any two positive integers $m$ and $n$. Let $l := mn$. Then $B_l$ is a disjoint union of $m^d$ translates of $B_n$. Let $\mb$ denote this set of translates of $B_n$.  For any $U\in U(B_l)$ and any $B\in \mb$, let $U_B$ denote the restriction of the configuration $U$ to the box $B$. Note that $U_B\in U(B)$.  

Since the eigenvalues of a unitary matrix lie on the unit circle, the range of the function $\phi$ is contained in the interval $[0,2N]$. This implies that for any $U\in U(B_l)$,
\begin{align}\label{hbl}
0\le S_{B_l}(U) - \sum_{B\in \mb} S_B(U_B) &\le Cm^d n^{d-1}\,,
\end{align} 
since the number of plaquettes of $B_l$ that are not plaquettes of any element of $\mb$ is bounded by a constant times $m^d n^{d-1}$. Since 
\begin{align*}
\int_{U(B_l)} \exp\biggl(-\beta \sum_{B\in \mb} S_B(U_B)\biggr)\, d\sigma_{B_l}(U) &= \prod_{B\in \mb} \int_{U(B)} \exp(-\beta S_B(U))\, d\sigma_B (U)\\
&= Z(B_n, g_0)^{|\mb|} = Z(B_n, g_0)^{m^d}\,,
\end{align*}
the inequality \eqref{hbl} implies that 
\begin{equation}\label{zbl}
Z(B_n,g_0)^{m^d} e^{-C\beta m^d n^{d-1}}\le Z(B_l,g_0)\le Z(B_n,g_0)^{m^d}\,.
\end{equation}
Taking logarithm on both sides and dividing by $l^d$ gives
\[
F(B_n, g_0) - \frac{C\beta}{n}\le F(B_l,g_0)\le F(B_n, g_0)\,.
\] 
Therefore,
\[
|F(B_n,g_0)-F(B_l,g_0)|\le \frac{C\beta }{n}\,.
\]
Exchanging $m$ and $n$ in the above argument, we see that the following inequality also holds:
\[
|F(B_m,g_0)-F(B_l,g_0)|\le \frac{C\beta }{m}\,.
\]
The proof is now easily completed by combining the last two displays.
\end{proof}
The next lemma establishes the required lower bound if $n$ is not growing too fast with $\beta$. This lemma heavily uses the estimates for lattice Maxwell theory that we obtained earlier. 
\begin{lmm}\label{downlmm2}
If $\beta$ is large enough (depending on $N$ and $d$) and $n\le \beta^2$, then
\begin{align*}
F(B_n, g_0) &\ge \frac{|E_n^1|}{n^d}\log C_N -\frac{|E_n^1|}{2n^d}N^2\log \beta + N^2 F_M(B_n) - \frac{C\log \beta}{\beta^{1/(10d+20)}}- \frac{C\log n}{\sqrt{n}}\,,
\end{align*}
where $C$ is a positive constant that depends only on $N$ and $d$.
\end{lmm}
\begin{proof}
Take any $n\le \beta^2$. Recall the constant $r_0$ from Theorem \ref{liethm}. Take any $c\in (1/3,1/2)$ and let $r:= \beta^{-c}$. Assume that $\beta$ is so large that $r\in (0,r_0]$. Observe that by Corollary~\ref{uu0},
\begin{align}
Z(B_n, g_0) &= \int_{U_0(B_n)}  e^{-\beta S_{B_n}(U)}\, d\sigma^0_{B_n}(U)  \nonumber\\
&\ge \int_{A}  e^{-\beta S_{B_n}(U)} \,d\sigma^0_{B_n}(U)\,,\label{lowstep1}
\end{align} 
where
\[
A:= \{U\in U_0(B_n): U(x,y)\in B(I, 2r)\text{ for all } (x,y)\in E_n\}\,.
\]
Let $T$ be the subset of $H_0(B_n)$ where $H(x,y)\in b(0,r)$ for each $(x,y)\in E_n^1$. Then by the lower bound from Theorem~\ref{liethm}, 
\begin{align*}
\int_{A}  e^{-\beta S_{B_n}(U)}\, d\sigma^0_{B_n}(U)&\ge C_N^{|E_n^1|}e^{-Crn^d}\int_{T}e^{-\beta S_{B_n}(e^{\I H})}\, d\lambda_{n}^0(H)\,.
\end{align*}
But by Lemma \ref{utohlmm}, for any $H\in T$,
\[
\biggl|S_{B_n}(e^{\I H}) - \frac{1}{2}M_n(H)\biggr|\le Cr^3n^d\,.
\]
Thus,
\begin{align}
\int_{A}  e^{-\beta S_{B_n}(U)}\, d\sigma^0_{B_n}(U) &\ge C_N^{|E_n^1|}e^{-C(r+\beta r^3 )n^d}\int_{T}e^{-\frac{1}{2}\beta M_n(H)}\, d\lambda_{n}^0(H)\,. \label{lowstep2}
\end{align}
Making the change of variable $G=\sqrt{\beta} H$, we get
\begin{align}
\int_{T}e^{-\beta M_n(H)} \,d\lambda_{n}^0(H) &= \beta^{-\frac{1}{2}N^2 |E_n^1|}\int_{T'}e^{-\frac{1}{2}M_n(G)}\, d\lambda_{n}^0(G)\,.\label{lowstep3}
\end{align}
where $T'$ is the set of all $G\in H(B_n)$ such that $G(x,y)\in b(0,\sqrt{\beta} r)$ for each $(x,y)$. Let
\[
\eta := \min\biggl\{ \frac{\sqrt{\beta} r}{2N}\, , \, C_1 n^{(d+2)/2}\biggr\}\,,
\]
where $C_1$ is the constant from Theorem \ref{latticemain}. 
Choose an integer $m$ such that 
\begin{equation}\label{c1kd}
\frac{\eta}{2}\le C_1m^{d+2}\le \eta\,.
\end{equation}
 It is possible to choose such an $m$ if $n$ and $\beta$ are large enough, since $r=\beta^{-c}$ and $c<1/2$. Let $T''$ be the set of all $G\in H(B_n)$ such that for any $(x,y)\in E_n$, the real and imaginary parts of every entry of the matrix $G(x,y)$ have magnitudes $\le C_1m^{d+2}$. By \eqref{c1kd}, it follows that  
\[
T''\subseteq T'\,.
\] 
Also, by \eqref{c1kd}, $m\le \sqrt{n}$. Therefore by the lower bound from \eqref{c1kd}, the probability inequality from Theorem~\ref{latticemain}, and the product structures of the set $T''$ and the measure $e^{-\frac{1}{2}M_n(G)}\, d\lambda_{n}^0(G)$, it follows that
\begin{align}
\frac{1}{Z_H(B_n)}\int_{T'}e^{-\frac{1}{2}M_n(G)}\, d\lambda_{n}^0(G)&\ge \frac{1}{Z_H(B_n)}\int_{T''}e^{-\frac{1}{2}M_n(G)}\, d\lambda_{n}^0(G) \nonumber\\
&\ge \exp\biggl(-\frac{Cn^d \log n}{m}\biggr)\nonumber\\
&\ge \exp\biggl(-\frac{Cn^d \log n}{(\sqrt{\beta} r)^{1/(d+2)}}-\frac{Cn^d \log n}{\sqrt{n}}\biggr)\,.\label{lowstep4}
\end{align}
Combining \eqref{lowstep1}, \eqref{lowstep2}, \eqref{lowstep3} and \eqref{lowstep4}, and applying \eqref{zhzm}, we get
\begin{align*}
F(B_n, g_0) &\ge \frac{|E_n^1|}{n^d}\log C_N -\frac{|E_n^1|}{2n^d}N^2\log \beta + N^2 F_M(B_n) \\
&\qquad  - C(r +\beta r^3) - \frac{C\log n}{(\sqrt{\beta} r)^{1/(d+2)}} - \frac{C\log n}{\sqrt{n}}\,.
\end{align*}
Since $n\le \beta^2$ and $r = \beta^{-c}$ for some $c\in (1/3,1/2)$, this gives
\begin{align*}
F(B_n, g_0) &\ge \frac{|E_n^1|}{n^d}\log C_N -\frac{|E_n^1|}{2n^d}N^2\log \beta + N^2 F_M(B_n)\\
&\qquad  -  C(\beta^{-c} +\beta^{1-3c}) - \frac{C\log \beta}{\beta^{(1-2c)/2(d+2)}}- \frac{C\log n}{\sqrt{n}}\,.
\end{align*}
The proof is completed by taking $c = 2/5$ and observing that the fourth term on the right is dominated by the fifth.
\end{proof}
Finally, the following lemma combines Lemma \ref{downlmm1} and Lemma \ref{downlmm2} to remove all constraints on the growth rates of $n$ and $\beta$.
\begin{lmm}\label{downlmmmain}
\begin{align*}
\liminf_{\substack{n\ra\infty\\ g_0\ra 0}}\biggl(F(B_n, g_0) + \frac{|E_n^1|}{2n^d}N^2\log \beta \biggr) \ge (d-1)\log C_N + N^2 \lim_{n\ra\infty} F_M(B_n)\,.
\end{align*}
\end{lmm}
\begin{proof}
If $n$ grows slower than $\beta^2$, then the claim follows by Lemma \ref{downlmm2}, Lemma \ref{en1lmm} and Theorem~\ref{cdthm}. Suppose that $n$ grows faster than $\beta^2$. Let $m$ be the integer part of $\beta^2$. Then by Lemma~\ref{downlmm1}, 
\begin{align*}
&F(B_n, g_0) + \frac{|E_n^1|}{2n^d}N^2\log \beta \ge F(B_m, g_0) - \frac{C}{\beta} + \frac{|E_n^1|}{2n^d}N^2\log \beta \\
&= \biggl(F(B_m, g_0) + \frac{|E_m^1|}{2m^d}N^2\log \beta\biggr)-\frac{C}{\beta}+\biggl(\frac{|E_n^1|}{2n^d} - \frac{|E_m^1|}{2m^d}\biggr)N^2\log \beta \,.
\end{align*}
We have already argued that the $\liminf$ of the first term has the desired lower bound, since $m$ is growing like $\beta^2$. To complete the proof, note that since $n$ and $m$ are both growing at least as fast as $\beta^2$, Lemma~\ref{en1lmm} ensures that the third term tends to zero.
\end{proof}
We are now ready to complete the proof of Theorem \ref{mainthm}.
\begin{proof}[Proof of Theorem \ref{mainthm}]
By Lemma \ref{uplmmmain} and Lemma \ref{downlmmmain}, we get
\[
\lim_{\substack{n\ra\infty\\g_0\ra 0}} \biggl(F(B_n, g_0) + \frac{|E_n^1|}{2n^d}N^2\log \beta\biggr) = (d-1)\log C_N+ N^2\lim_{n\ra\infty} F_M(B_n)\,.
\]
Lemma \ref{en1lmm} gives the size of $E_n^1$. On the other hand by \eqref{fmkn}, Theorem \ref{cdthm} and Lemma \ref{en1lmm},
\[
\lim_{n\ra\infty}F_M(B_n) = K_{d} + \frac{d-1}{2}\log (2\pi)\,.
\]
Lastly, note that 
\begin{align*}
\log C_N + \frac{N^2}{2}\log(2\pi) = \log \biggl( \frac{\prod_{j=1}^{N-1} j!}{(2\pi)^{N/2}}\biggr)\,.
\end{align*}
This completes the proof of Theorem \ref{mainthm}.
\end{proof}

\section*{Acknowledgments} 
I thank  Erik Bates, Persi Diaconis, Bruce Driver, Alex Dunlap,  Jafar Jafarov, Todd Kemp, Tim Nguyen, Erhard Seiler, Ambar Sengupta, Steve Shenker, Lenny Susskind and Mithat \"Unsal for helpful discussions and communications.  I am particularly indebted to David Brydges for clarifying many aspects of the literature on constructive quantum field theory, and to Len Gross for a long list of useful comments.

\bibliographystyle{plainnat}

\end{document}